\documentclass[aps,reprint,superscriptaddress,showpacs]{revtex4-2}

\usepackage[utf8]{inputenc}
\usepackage{amsmath}
\usepackage{amsthm}
\usepackage{amsfonts}
\usepackage{amssymb}
\usepackage{mathrsfs}
\usepackage{graphicx}
\usepackage{enumerate}
\usepackage{tabularx}
\usepackage{tensor}
\usepackage{bm}

\newtheorem{theorem}{Theorem}[section]
\newtheorem{lemma}{Lemma}[section]

\begin{document}

\title{Higher-dimensional minimal theory of mass-varying massive gravity and its cosmological consequences}

\author{Ahmad Khoirul Falah}
\email{akhoirulfalah94@students.itb.ac.id}
\affiliation{Theoretical Physics Laboratory, Theoretical High Energy Physics Research Division, Faculty of Mathematics and Natural Sciences, Institut Teknologi Bandung, Jl.~Ganesha no.~10 Bandung, Indonesia, 40132}

\author{Andy Octavian Latief}
\email{latief@itb.ac.id}
\affiliation{Physics of Magnetism and Photonics Research Division, Faculty of Mathematics and Natural Sciences, Institut Teknologi Bandung, Jl.~Ganesha no.~10 Bandung, Indonesia, 40132}

\author{Husin Alatas}
\email{alatas@apps.ipb.ac.id}
\affiliation{Theoretical Physics Division, Department of Physics, IPB University (Bogor Agricultural University), Jl.~Meranti, Kampus IPB Darmaga, Bogor 16680, Indonesia}

\author{Bobby Eka Gunara}
\email{bobby@fi.itb.ac.id (Corresponding author)}
\affiliation{Theoretical Physics Laboratory, Theoretical High Energy Physics Research Division, Faculty of Mathematics and Natural Sciences, Institut Teknologi Bandung, Jl.~Ganesha no.~10 Bandung, Indonesia, 40132}

\begin{abstract}
    In this paper we construct higher-dimensional minimal theory of mass-varying massive gravity (MTMVMG) where the masslike scalar potential is coupled to a vielbein potential, unlike in the previous literature where it is coupled to metric, such that the number of graviton degrees of freedom in the theory is the same as in general relativity. We then study the cosmological aspects of this theory and show that it has eight critical points: five in the massless sector and three in the massive sector. In contrast to the standard theory of mass-varying massive gravity where the graviton mass asymptotically approaches zero at late times, hence making the contribution of massive gravity to the late-time cosmic expansion minimum, the MTMVMG can provide good descriptions both in the massless and massive sectors. Especially, there are at least two interesting possible scenarios for the late-time cosmology in the theory: the dark energy is either due to the constant graviton mass which comes from the scalar field that becomes frozen after the reheating era, or due to the quintessence paradigm where the scalar field is dynamic. Therefore, if the accelerating expansion of the universe in the massless sector can be explained by standard quintessence paradigm, in the massive sector it has to be explained by the nontrivial interplay between quintessence and massive gravity.
\end{abstract}

\maketitle

\section{Introduction}

\label{sec:intro}

Several alternative theories modifying Einstein's general relativity have been proposed in the past decades as parts of an effort to solve long-standing cosmological problems such as dark energy, dark matter, and inflation. One of these theories, based on the assumption that graviton might have nonzero mass and hence later coined as the theory of \textit{massive gravity}, was originally visioned by M.~Fierz and W.~Pauli in 1939 \cite{fierz1939relativistic}. However, the theory did not have a continuous transition to general relativity in the limit of zero graviton mass, an issue known as the van Dam--Veltman--Zakharov discontinuity \cite{vandam1970massive,zakharov1970linearized}. This problem was remedied by Vainshtein's nonlinear mechanism in 1972 \cite{vainshtein1972problem}, but the nonlinear terms then gave rise to another problem called the Boulware-Deser ghost \cite{boulware1972can}. Building on several previous attempts \cite{arkanihamed2003effective,creminelli2005ghosts}, this problem was then finally resolved in 2010 by C.~de Rham, G.~Gabadadze, and A.~Tolley (dRGT) \cite{derham2010generalization,derham2011resummation}, resulting in a Lorentz-invariant, ghost-free nonlinear theory of massive gravity \cite{hassan2012resolving,hassan2012ghost,derham2012ghost,derham2011helicity,hassan2012confirmation,mirbabayi2012proof,golovnev2012on,hassan2012proof,kluson2012nonlinear}. (See also Refs.~\cite{hinterbichler2012theoretical,derham2014massive} for reviews.)

Unfortunately, the dRGT theory also had some serious challenges: there was no stable homogeneous and isotropic cosmological solutions \cite{damico2011massive,gumrukccuouglu2011open,gumrukcuoglu2012cosmological,defelice2012massive}, together with other pathologies such as Higuchi bound \cite{higuchi1987forbidden,fasiello2012cosmological} and positivity bound \cite{cheung2016positive,bellazzini2018beyond}. To overcome this, the \textit{minimal theory of massive gravity} (MTMG) was then proposed by A.~De Felice and S.~Mukohyama in 2016 \cite{defelice2016minimal} by imposing some constraints which suppress the five degrees of freedom in the original dRGT theory such that there are only two degrees of freedom, both of them are tensor modes, as in the case of general relativity, but now the theory is not Lorentz invariant. It has the same Friedmann--Lema\^{\i}tre--Robertson--Walker (FLRW) equations as in the dRGT theory, but now the FLRW background is stable \cite{defelice2016phenomenology}. There are two branches of solutions. The first is the self-accelerating branch, which is phenomenologically the same with the $\Lambda$CDM cosmology, except that the accelerating expansion of the universe is now caused by the graviton mass term, not necessarily by the cosmological constant. The second is the normal branch, which is phenomenologically different from general relativity in the scalar and tensor sectors, leading to nontrivial dynamics which could be tested against the predictions of general relativity \cite{defelice2016phenomenology,defelice2017graviton}.

Another attempt to modify the dRGT theory was done by Q.-G.~Huang, Y.-S.~Piao, and S.-Y.~Zhou in 2012 by coupling the graviton potentials to a scalar field $\psi$, enabling the graviton to have a varying mass \cite{huang2012mass}. This theory of \textit{mass-varying massive gravity} (MVMG) is again Boulware-Deser ghost-free and the Lorentz invariance is satisfied. Varying the graviton mass can lead to interesting cosmological behaviors both in the inflationary and late-time era. Specifically, the graviton mass will asymptotically approach zero at late times due to the dynamics of the theory, hence there is no need to fine tune the graviton mass to a very small number to be in line with the cosmological bounds for the graviton mass in the present time \cite{saridakis2013phantom}. However, this may be a disadvantage, since it means that the contribution of massive gravity to explain the cosmic expansion at late times is minimum \cite{tannukij2016mass}.

Another disadvantage of the MVMG theory is that it has many graviton degrees of freedom, which may lead to instabilities. Therefore, to suppress these degrees of freedom, we will follow in this paper the method of Refs.~\cite{defelice2016minimal,defelice2016phenomenology}. First, we define the precursor theory by writing the MVMG action using the vielbein formulation in the Arnowitt--Deser--Misner (ADM) formalism, but here we generalize the theory to the case of higher dimensions. We also adopt the vielbein potential in Ref.~\cite{hinterbichler2012interacting} and couple it to the masslike scalar potential $W(\psi)$, unlike in the previous literature where it is coupled to metric. The purpose of this is such that the number of graviton degrees of freedom in the theory is the same as in general relativity. We then perform the Legendre transformation to the precursor action to obtain the precursor Hamiltonian. After imposing the nontrivial constraints to the theory, we will obtain the \textit{minimal theory of mass-varying massive gravity} (MTMVMG), where the number of graviton degrees of freedom in this theory becomes ${D (D - 3)}/2$, as in the $D$-dimensional general relativity, where $D$ is the number of spacetime dimensions. For $D = 4$, the number of graviton degrees of freedom in the MTMVMG is two, in contrast to the MVMG theory where there are five graviton degrees of freedom. Therefore, in the light of \textit{minimally modified gravity} (MMG), a modified theory of gravity with two local gravitational degrees of freedom, discussed in Refs.~\cite{lin2017class,mukohyama2019minimally}, where the \textit{type-I} MMG is for the theory in which there exists an Einstein frame and the \textit{type-II} MMG is for the one in which there is no Einstein frame \cite{aoki2019phenomenology,defelice2020theory}, the MTMVMG can be viewed as an \textit{extended} type-II MMG theory.

To study the cosmological aspects of this theory, we can perform the Legendre transformation again to get the expression for the MTMVMG action, which then can be used to obtain the Friedmann-Lema\^{\i}tre equations. We take both the scalar potential and the graviton mass couplings to have exponential forms, and find that there are eight critical points in the theory: five in the massless sector and three in the massive sector. Therefore, the MTMVMG theory can have both massless and massive sectors even in the late-time era, in contrast to the ordinary MVMG theory where, as mentioned previously, the dynamics of the theory will lead the graviton mass to asymptotically approach zero at late times. This makes the MTMVMG a richer theory, which then can give us good descriptions of both the inflationary and late-time era. Especially, there are at least two interesting possible scenarios for the late-time cosmology: the dark energy is either due to the constant graviton mass which comes from the scalar field $\psi_\infty$ that becomes frozen after the reheating era, or due to the quintessence paradigm where the scalar field $\psi$ is dynamic. Therefore, if the accelerating expansion of the universe in the massless sector can be explained by standard quintessence paradigm, in the massive sector it has to be explained by the nontrivial interplay between quintessence and massive gravity.

This paper is organized as the following: In Sec.~\ref{sec:precursor} we derive the precursor action by writing the MVMG theory using the vielbein formulation. In Sec.~\ref{sec:MVMGS} we then construct the minimal theory by imposing $D$-constraints. In Sec.~\ref{sec:Friedmanneq} we derive the Friedmann-Lema\^{\i}tre equations for our model. We then perform in Sec.~\ref{sec:Dynas} the dynamical analysis around the critical points and check in Sec.~\ref{sec:locglobex} their local and global existences. In Sec.~\ref{sec:cosmologicalcons} we discuss the cosmological implication on the inflationary expansion and late-time acceleration. We then conclude the paper and write several remarks in Sec.~\ref{sec:conclusions}. Detailed calculations are presented in the Appendix.

\section{Precursor Theory}

\label{sec:precursor}

In this section we will construct the MVMG action by replacing the graviton mass with the masslike scalar potential $W(\psi)$. However, we will adopt the vielbein potential from Ref.~\cite{hinterbichler2012interacting} and couple it to $W(\psi)$, in contrast to Ref.~\cite{huang2012mass} where $W(\psi)$ is coupled to metric. We will then define the action for the precursor theory, which will be needed later to construct the MTMVMG action.

Let us first consider two $D$-dimensional Lorentzian manifolds $(\mathcal{M}, g)$ and $(\mathcal{M}_0, g_0)$ parametrized by the coordinate systems $x^\mu$ and $y^a$ with $\mu, a = 0, \ldots, D - 1$, respectively. Using the ADM formalism, the metrics $g$ and $g_0$ can be written as
\begin{eqnarray}
    ds^2_{(d)} &=& -N^2 dt^2 + \gamma_{i j} (dx^i + N^i dt) (dx^j + N^j dt), \label{eq:metfoliasi1} \\
    ds^2_{(b)} &=& -N_0^2 d\tau^2 + \gamma_{0 i j} (dy^i + N_0^i d\tau) (dy^j + N_0^j d\tau), \label{eq:metfoliasi2}
\end{eqnarray}
where $\gamma_{i j}$ and $\gamma_{0 i j}$ are the components of the induced spatial metrices, $N$ and $N_0$ are the lapse functions, $N^i$ and $N_0^i$ are the shift vectors, and $i, j = 1, \ldots, D - 1$. The subscripts $(d)$ and $(b)$ denote the dynamic and the background, respectively. Let us then introduce the vielbeins $\tensor{E}{^A_\mu}$ and $\tensor{E}{_0^A_a}$ such that the metrics \eqref{eq:metfoliasi1} and \eqref{eq:metfoliasi2} can be written as $g_{\mu \nu} = \eta_{A B} \tensor{E}{^A_\mu} \tensor{E}{^B_\nu}$ and $g_{0 a b} = \eta_{A B} \tensor{E}{_0^A_a} \tensor{E}{_0^B_b}$, where $A, B = 0, \ldots, D - 1$ and $\eta_{A B}$ is the flat Minkowski metric. Therefore, the vielbeins $\tensor{E}{^A_\mu}$ and $\tensor{E}{_0^A_a}$ have the form
\begin{equation}
    \tensor{E}{^A_\mu} = \begin{pmatrix} N & 0 \\ N^k \tensor{e}{^I_k} & \tensor{e}{^I_i} \end{pmatrix}, \qquad \tensor{E}{_0^A_a} = \begin{pmatrix} N_0 & 0 \\ N_0^k \tensor{e}{_0^I_k} & \tensor{e}{_0^I_i} \end{pmatrix}, \label{eq:basis}
\end{equation}
whose duals are given by
\begin{equation}
    \tensor{E}{^\mu_A} = \begin{pmatrix} \frac{1}{N} & 0 \\ -\frac{N^i}{N} & \tensor{e}{^i_I} \end{pmatrix}, \qquad \tensor{E}{_0^a_A} = \begin{pmatrix} \frac{1}{N_0} & 0 \\ -\frac{N_0^i}{N_0} & \tensor{e}{_0^i_I} \end{pmatrix}, \label{eq:dualbasis}
\end{equation}
such that they satisfy
\begin{eqnarray*}
    \tensor{E}{^A_\mu} \tensor{E}{^\mu_B} = \tensor{\delta}{^A_B}, &\qquad& \tensor{E}{_0^A_a} \tensor{E}{_0^a_B} = \tensor{\delta}{^A_B}, \\
    \tensor{E}{^\mu_A} \tensor{E}{^A_\nu} = \tensor{\delta}{^\mu_\nu}, &\qquad& \tensor{E}{_0^a_A} \tensor{E}{_0^A_b} = \tensor{\delta}{^a_b}.
\end{eqnarray*}
Here $\tensor{e}{^I_i}$ and $\tensor{e}{_0^I_i}$ are the spatial vielbeins, with $i, I = 1, \ldots, D - 1$, and $\tensor{e}{^i_I}$ and $\tensor{e}{_0^i_I}$ are their duals, respectively.

Now, we consider a smooth embedding $\phi: \mathcal{M} \longrightarrow \mathcal{M}_0$ such that we could have pulled back quantities
\begin{eqnarray}
    \tensor{\tilde{E}}{^A_\mu}(x) &=& \frac{\partial \phi^a}{\partial x^\mu} \tensor{E}{_0^A_a} (\phi(x)), \label{eq:transfgaugepre1} \\
    f_{\mu \nu}(x) &=& \frac{\partial \phi^a}{\partial x^\mu} \frac{\partial \phi^b}{\partial x^\nu} g_{0 a b} (\phi(x)), \label{eq:transfgaugepre2}
\end{eqnarray}
so that $\mathcal{M}$ obeys diffeomorphism rules through a St\"uckelberg field $\phi^a$. In the unitary gauge where $\phi^a = \tensor{\delta}{^a_\mu }x^\mu$, we simply recover $\tensor{\tilde{E}}{^A_\mu} = \tensor{E}{_0^A_\mu}$ and $f_{\mu \nu} = g_{0 \mu \nu}$. Note that the original formulation of the dRGT theory has a Minkowski background, $g_{0 a b} = \eta_{a b}$. Using the vielbein formulation in Ref.~\cite{hinterbichler2012interacting}, the ghost-free potential related to the graviton mass has the form
\begin{equation}
    \sum_{n = 0}^D \frac{c_n}{n! (D - n)!} \hat{\epsilon}_{A_1 A_2 \cdots A_D} \tensor{\bm{E}}{^{A_1}} \wedge \cdots \wedge \tensor{\bm{E}}{^{A_n}} \wedge \tensor{\bm{\tilde{E}}}{^{A_{n + 1}}} \wedge \cdots \wedge \tensor{\bm{\tilde{E}}}{^{A_D}} \label{eq:gfpot}
\end{equation}
with one-forms $\bm{E}^A = \tensor{E}{^A_\mu} dx^\mu$ and $\bm{\tilde{E}}^A = \tensor{\tilde{E}}{^A_\mu} dx^\mu$ the dRGT mass term for arbitrary background and dimension. The quantity $\hat{\epsilon}$ denotes the Levi-Civita symbol in the flat spacetime. Furthermore, inspired by \cite{huang2012mass}, the graviton mass is replaced by a function of a scalar field $W(\psi)$, where $\psi(x)$ is well defined on $\mathcal{M}$. Coupling $W(\psi)$ to the potential in Eq.~\eqref{eq:gfpot} and adding this coupling term to the Einstein--Klein--Gordon action, we obtain the ghost-free MVMG action as the following,
\begin{widetext}
\begin{eqnarray}
    S_\text{MVMG} &=& \int_{\mathcal{M}} d^Dx \, \det{(E)} \, \left(\frac{M_{\text{Pl}}^{D - 2}}{2} R(E) - \frac{1}{2} \partial^\mu \psi \partial_\mu \psi - V(\psi) \right) \nonumber \\
    && - \, \frac{1}{4} \int_{\mathcal{M}} {W(\psi)} \left( \sum_{n = 0}^D \frac{c_n}{n! (D - n)!} \hat{\epsilon}_{A_1 A_2 \cdots A_D} \tensor{\bm{E}}{^{A_1}} \wedge \cdots \wedge \tensor{\bm{E}}{^{A_n}} \wedge \tensor{\bm{\tilde{E}}}{^{A_{n + 1}}} \wedge \cdots \wedge \tensor{\bm{\tilde{E}}}{^{A_D}} \right), \label{eq:ev1}
\end{eqnarray}
where $V(\psi)$ is the scalar potential function.

The action for the precursor theory can be obtained by simply substituting the vielbeins in \eqref{eq:basis} and \eqref{eq:dualbasis} to the action \eqref{eq:ev1},
\begin{eqnarray}
    S_{\text{pre}} &=& \int d^Dx \, N \det{(e)} \Bigg[ \frac{M_{\text{Pl}}^{D - 2}}{2} \left({}^{(D - 1)}R(e) + K^{i j} K_{i j} - K^2 \right) + \frac{1}{2 N^2} \dot{\psi}^2 - \frac{1}{2} \partial^i \psi \partial_i \psi + \frac{N^i N^j}{2N^2} \partial_i \psi \partial_j \psi - \frac{N^i}{N^2} \dot{\psi} \partial_i \psi - V(\psi) \nonumber \\
    && - \, W(\psi) \left( | \det{(X)} | \frac{M}{N} \sum_{n = 0}^{D - 1} c_n \mathcal{S}_n(Y) + \sum_{n = 0}^{D - 1} c_{D - n} \mathcal{S}_n(X) \right) \Bigg], \label{eq:Spre}
\end{eqnarray}
\end{widetext}
where ${}^{(D - 1)}R(e)$ is the spatial Ricci scalar, while $K^{i j}$ and $K$ are the second fundamental form and the mean curvature, respectively. For any function $\psi$, we define $\dot{\psi} \equiv {\partial \psi}/{\partial t}$ and $\partial_i \psi \equiv {\partial \psi}/{\partial x^i}$. Here, the quantities $\{M, M^i, \tensor{\tilde{e}}{^I_i}\}$ are the pull back of $\{N_0, N_0^i, \tensor{e}{_0^I_i}\}$ via St\"uckelberg field $\phi$ as a background image on $\mathcal{M}$ given by Eq.~\eqref{eq:transfgaugepre1} and \eqref{eq:transfgaugepre2}, respectively. The elements $\mathcal{S}_n$ are the $n$th order symmetric polynomials which depend on either $\tensor{Y}{^I_J} \equiv \tensor{e}{^I_k} \tensor{\tilde{e}}{^k_J}$ or $\tensor{X}{^I_J} \equiv \tensor{\tilde{e}}{^I_k} \tensor{e}{^k_J}$ (see Appendix \ref{sec:appendixA} for more discussions). Note that the action \eqref{eq:Spre} violates the local Lorentz symmetry because we have used the ADM vielbeins in Eqs.~\eqref{eq:basis} and \eqref{eq:dualbasis}.

\section{Minimal Theory of MVMG}

\label{sec:MVMGS}

The discussion in this section is divided into two parts. First, we construct the Hamiltonian for the MTMVMG and identify some constraints which restricts the graviton degrees of freedom. Second, we discuss the MTMVMG action which is constructed from the precursor action discussed in the previous section but with some additional constraints.

\subsection{MTMVMG Hamiltonian and some constraints}

\label{sec:Hamil}

Let us first consider the spatial vielbeins $\tensor{e}{^I_i}$ and the scalar field $\psi$ as the canonical variables which correspond to the conjugate momenta defined as
\begin{eqnarray}
    \tensor{\pi}{^i_I} &\equiv& \frac{\delta S_{\text{pre}}}{\delta \tensor{\dot{e}}{^I_i}} \nonumber \\
    &=& \det{(e)} \, M_{\text{Pl}}^{D - 2} (K^{i j} - K \gamma^{ij}) \delta_{I J} \tensor{e}{^J_j}, \label{eq:pikon1} \\
    \pi &\equiv& \frac{\delta S_{\text{pre}}}{\delta \dot{\psi}} \nonumber \\
    &=& \det{(e)} \, \left( \frac{1}{N} \dot{\psi} - \frac{N^i}{N} \partial_i \psi \right). \label{eq:pikon2}
\end{eqnarray}
We can switch from Lagrangian to Hamiltonian by performing the Legendre transformation in order to see some constraints of the theory. As it is well-known in the vielbein language that the lapse function $N$ and the shift vector $N^i$ appear as Lagrange multipliers enforcing the diffeomorphism constraints \cite{hinterbichler2012interacting}, namely $\mathcal{R}_0 \approx 0$ and $\mathcal{R}_i \approx 0$. These constraints are called the primary constraints of the first kind, which then enable us to construct another set of constraints called the secondary constraints of the first kind \cite{henneaux1994quantization}. Note that there are only $D - 2$ independent secondary constraints, since two of them can be obtained from the others. We denote them as $\tilde{\mathcal{C}}_\tau$ $(\tau = 1, \ldots, D - 2)$, together with their Lagrange multipliers $\lambda^\tau$.

Additionally, as studied in Refs.~\cite{defelice2016minimal,defelice2016phenomenology}, from Eq.~\eqref{eq:pikon1} we also have another set of primary constraints of the second kind $\mathcal{P}^{[M N]}$ that lead to the secondary constraints of the second kind $\mathcal{Z}^{[M N]}$ in the phase space, together with their Lagrange multipliers $\alpha_{M N}$ and $\beta_{M N}$, with $M, N = 1, \ldots, (D - 1) (D - 2)/2$. These secondary constraints are necessary since the primary constraints should be preserved with respect to the time evolution. Therefore, the precursor Hamiltonian can be written down as
\begin{eqnarray}
    H_{\text{pre}} &=& \int d^{D - 1}x \, \Big( -N \mathcal{R}_0 - N^i \mathcal{R}_i \nonumber \\
    && + \, W(\psi) N \mathcal{H}_0 + W(\psi) M \mathcal{H}_1 + \tilde{\lambda}^{\tau} \tilde{\mathcal{C}}_\tau \nonumber \\
    && + \, \alpha_{M N} \mathcal{P}^{[M N]} + \beta_{M N} \mathcal{Z}^{[M N]} \Big), \label{eq:hampre}
\end{eqnarray}
where
\begin{eqnarray}
    \mathcal{R}_0 &\equiv& \det{(e)} \, \frac{M_{\text{Pl}}^{D - 2}}{2} {}^{(D - 1)}R(e) \nonumber \\
    && - \, \frac{1}{2 \det{(e)} \, M_{\text{Pl}}^{D - 2}} \left[ \tensor{\pi}{^i_I} \tensor{\pi}{^I_i} - \frac{1}{D - 2} (\tensor{\pi}{^i_I} \tensor{e}{^I_i})^2 \right], \nonumber \\
    && - \, \frac{1}{2 \det{(e)}} \pi^2 - \frac{\det{(e)}}{2} \partial_i \psi \partial^i \psi \nonumber \\
    && - \, \det{(e)} \, V(\psi), \\
    \mathcal{R}_i &\equiv& \nabla_j (\tensor{\pi}{^j_I} \tensor{e}{^I_i}) - \pi \partial_i \psi, \\
    \mathcal{H}_0 &\equiv& \det{(e)} \, \sum_{n = 0}^{D - 1} c_{D - n} \mathcal{S}_n(X), \\
    \mathcal{H}_1 &\equiv& \det{(e)} \, | \det{(X)} | \sum_{n = 0}^{D - 1} c_n \mathcal{S}_n(Y),
\end{eqnarray}
and
\begin{eqnarray}
    \mathcal{P}^{[M N]} &\equiv& (\tensor{e}{^M_j} \delta^{K N} - \tensor{e}{^N_j} \delta^{K M}) \tensor{\pi}{^j_K}, \\
    \mathcal{Z}^{[M N]} &\equiv& (\tensor{e}{^M_j} \delta^{K N} - \tensor{e}{^N_j} \delta^{K M}) \tensor{\tilde{e}}{^j_K}.
\end{eqnarray}
Since the constraints above remove some graviton degrees of freedom, we can construct a theory in which the spatial graviton degrees of freedom coincide with the standard general relativity. Inspired by Ref.~\cite{defelice2016minimal}, we can impose the $D$-constraints in unitary gauge given by
\begin{eqnarray}
    \mathcal{C}_0 &\equiv& \{\mathcal{R}_0, H_{1}\}_\text{PB} - W(\psi) \frac{\partial \mathcal{H}_0}{\partial t} \approx 0, \label{Dconstraints1} \\
    \mathcal{C}_i &\equiv& \{\mathcal{R}_i, H_{1}\}_\text{PB} \approx 0, \label{Dconstraints2}
\end{eqnarray}
where $\{ \cdots \}_\text{PB}$ denotes the Poisson bracket and
\begin{equation}
    H_1 \equiv \int d^{D - 1} x \, W(\psi) M \mathcal{H}_1.
\end{equation}
Note that these constraints consist of two new constraints and $D-2$ independent constraints $\tilde{\mathcal{C}}_\tau$ which already exist in the precursor theory. Moreover, the constraints in Eqs.~\eqref{Dconstraints1} and \eqref{Dconstraints2} imply that the theory admits the Lorentz symmetry violation. The Hamiltonian of the MTMVMG theory then reads
\begin{eqnarray}
    H_\text{MTMVMG} &=& \int d^{D - 1} x \, \Big[ -N \mathcal{R}_0 - N^i \mathcal{R}_i \nonumber \\
    && + \, W(\psi) (N \mathcal{H}_0 + M \mathcal{H}_1) + \lambda \mathcal{C}_0 + \lambda^i \mathcal{C}_i \nonumber \\
    && + \, \alpha_{M N} \mathcal{P}^{[M N]} + \beta_{M N} \mathcal{Z}^{[M N]} \Big]. \label{eq:hmtmg}
\end{eqnarray}
Thus, in total we have $D^2 - D + 2$ constraints, which means that the number of spatial graviton degrees of freedom are $D (D - 3)/2$, as in the $D$-dimensional general relativity.

\subsection{MTMVMG action}

To construct the MTMVMG action, one has to employ the Legendre transformation on the Hamiltonian functional \eqref{eq:hmtmg}. It will be shortly discussed in this subsection, but its detailed derivations will be presented in the Appendix \ref{sec:appendixB}.

As discussed in the previous subsection, we should have the nontrivial $D$-constraints \eqref{Dconstraints1} and \eqref{Dconstraints2} in the MTMVMG theory. In order to have a consistent theory, we have to modify the conjugate momenta \eqref{eq:pikon1} and \eqref{eq:pikon2} to
\begin{eqnarray}
    \frac{\tensor{\pi}{^i_I}}{\det{(e)}} &\equiv& M_{\text{Pl}}^{D - 2} (K^{i j} \delta_{I J} \tensor{e}{^J_j} - K \tensor{e}{^i_I}) \nonumber \\
    && - \, \lambda\frac{W(\psi)}{2} \frac{M}{N} \Theta^{i j} \delta_{I J} \tensor{e}{^J_j}, \label{eq:pikon3} \\
    \frac{\pi}{\det{(e)}} &\equiv& \frac{\dot{\psi}}{N} - \frac{N^i}{N} \partial_i \psi - \lambda \frac{dW}{d\psi} \frac{M}{N} \Phi, \label{pikonpi}
\end{eqnarray}
with
\begin{eqnarray}
    \Theta^{i j} &\equiv& - | \det{(X)} | \delta^{I K} (\tensor{e}{^i_K} \tensor{\tilde{e}}{^j_J} + \tensor{e}{^j_K} \tensor{\tilde{e}}{^i_J}) \nonumber \\
    && \times \, \sum_{n = 1}^{D - 1} \sum_{m = 1}^n (-1)^m c_n \tensor{\left(Y^{m - 1} \right)}{^J_I} \mathcal{S}_{n - m}(Y), \\
    \Phi &\equiv& | \det{(X)} | \sum_{n = 1}^{D - 1} c_n \mathcal{S}_n(Y).
\end{eqnarray}
The notation $\tensor{(M^m)}{^I_J}$ means
\begin{equation}
\tensor{(M^m)}{^I_J} \equiv \tensor{M}{^I_{K_1}} \tensor{M}{^{K_1}_{K_2}} \cdots \tensor{M}{^{K_{m - 1}}_J}.
\end{equation}
The modifications \eqref{eq:pikon3} and \eqref{pikonpi} imply that the MTMVMG theory modifies both the kinetic part and the mass term. As we will see later, this also provides a class of solutions which coincides in the dRGT theory in the FLRW background.

For the sake of convenience, let us first introduce the following tensors
\begin{equation}
    \tensor{\mathcal{K}}{^i_j} \equiv \tensor{\tilde{e}}{^i_I} \tensor{e}{^I_j}, \qquad \tensor{\mathcal{\bar{K}}}{^i_j} \equiv \tensor{e}{^i_I} \tensor{\tilde{e}}{^I_j}.
\end{equation}
satisfying $\tensor{\mathcal{K}}{^i_k} \tensor{\mathcal{\bar{K}}}{^k_j} = \tensor{\delta}{^i_j}$, which correspond to the spatial metrics by
\begin{equation}
    \tensor{\mathcal{K}}{^i_k} \tensor{\mathcal{K}}{^k_j} = \tilde{\gamma}^{i l} \gamma_{l j}, \qquad \tensor{\mathcal{\bar{K}}}{^i_k} \tensor{\mathcal{\bar{K}}}{^k_j} = \gamma^{i l} \tilde{\gamma}_{l j},
\end{equation}
where $\tilde{\gamma}_{i j} = \delta_{I J} \tensor{\tilde{e}}{^I_i} \tensor{\tilde{e}}{^J_j}$ is the spatial metric on $\mathcal{M}$. Performing the Legendre transformation, we obtain the MTMVMG action,
\begin{equation}
    S_{\text{MTMVMG}} = S_{\text{pre}} + S_\lambda, \label{eq:aksimtmg}
\end{equation}
where
\begin{widetext}
\begin{eqnarray}
    S_{\lambda} &=& \frac{2}{M_{\text{Pl}}^{D - 2}} \int d^D x \, N \sqrt{\gamma} \left( \lambda \frac{W(\psi)}{4} \frac{M}{N} \right)^2 \left( \gamma_{i k} \gamma_{j l} - \frac{1}{D - 2} \gamma_{i j} \gamma_{k l} \right) \Theta^{i j} \Theta^{k l} \nonumber \\
    && + \, \frac{1}{2} \int d^D x \, N \sqrt{\gamma} \left( \lambda \frac{dW}{d\psi} \frac{M}{N} \right)^2 \Phi^2 - \int d^D x \, \sqrt{\gamma} \left[ \lambda \mathcal{\bar{C}}_0 + \lambda^i \tensor{\mathcal{C}}{_i}\right],
\end{eqnarray}
and
\begin{eqnarray}
    \mathcal{\bar{C}}_0 &=& \frac{1}{2} W(\psi) M \left( \gamma_{i k} \gamma_{j l} - \frac{1}{D - 2} \gamma_{i j} \gamma_{k l} \right) \Theta^{k l} (K^{i j} - K \gamma^{i j}) \nonumber \\
    && + \, W(\psi) | \det{(\mathcal{\bar{K}})} | \sum_{n = 1}^{D - 1} \sum_{m = 1}^n (-1)^m c_n \tensor{\left( \mathcal{K}^{m - 1} \right)}{^k_l} \tensor{\tilde{\zeta}}{^l_k} \mathcal{S}_{n - m}(\mathcal{K}) \nonumber \\
    && - \, M | \det{(\mathcal{\bar{K}})} | \frac{dW}{d\psi} \left( \frac{\dot{\psi}}{N} - \frac{N^i}{N} \partial_i \psi \right) \sum_{n = 1}^{D - 1} c_n \mathcal{S}_n(\mathcal{K}), \\
    \tensor{\mathcal{C}}{_i} &=& - W(\psi) \nabla_k M | \det{(\mathcal{\bar{K}})} | \sum_{n = 1}^{D - 1} \sum_{m = 1}^n (-1)^m c_n \tensor{\left( \mathcal{K}^m \right)}{^k_i} \mathcal{S}_{n - m}(\mathcal{K}) \nonumber \\
    && - \, M | \det{(\mathcal{\bar{K}})} | \partial_i \psi \frac{dW}{d\psi} \sum_{n = 1}^{D - 1} c_n \mathcal{S}_n(\mathcal{K}),
\end{eqnarray}
with
\begin{eqnarray}
    \Theta^{i j} &=& - 2 | \det{(\mathcal{\bar{K}})} | \gamma^{i l} \sum_{n = 1}^{D - 1} \sum_{m = 1}^n (-1)^m c_n \tensor{\left(\mathcal{K}^{m}\right)}{^j_l} \mathcal{S}_{n - m}(\mathcal{K}) \label{eq:thetai} \\
    \Phi &=& | \det{(\mathcal{\bar{K}})} | \sum_{n = 1}^{D - 1} c_n \mathcal{S}_n(\mathcal{K}). \label{eq:phii}
\end{eqnarray}
\end{widetext}
It is worth mentioning that the additional term in \eqref{eq:pikon3} implies that we have to set $\alpha_{M N} = \beta_{M N} = 0$ \cite{defelice2016phenomenology,defelice2017minimal}. This is so because the tensors $\mathcal{P}^{[M N]}$ and $Y^{[M N]}$ are antisymmetric, while the tensor $\Theta^{i j}$ is symmetric.

We could extend the action \eqref{eq:aksimtmg} by adding the matter field,
\begin{equation}
    S_{\text{MTMVMG-M}} = S_{\text{pre}} + S_\lambda + S_\text{matter}. \label{eq:aksimtmg1}
\end{equation}
Here, we consider the matter field part $S_\text{matter}$ to be the perfect fluid whose energy-momentum tensor has the form
\begin{equation}
    T_{\mu \nu} = \frac{2}{\sqrt{-g}} \frac{\delta S_\text{matter}}{\delta g^{\mu \nu}} = \rho_m U_\mu U_\nu + P_m \left( g_{\mu \nu} + U_\mu U_\nu \right),
\end{equation}
where $U^\mu$ and $\rho_m$ are the unit velocity of the fluid and the energy density, respectively. The pressure $P_m$ is given by the state equation of matter fields,
\begin{equation}
    P_m = w_m \rho_m, \label{eq:state_eq}
\end{equation}
with $w_m$ is a real constant \cite{akbar2019local}. In the standard higher-dimensional cosmology, we particularly have $w_m = \frac{1}{D - 1}$ (radiation), $w_m = 0$ (dust), and $w_m = -1$ (vacuum) \cite{chatterjee1990homogeneous}.

\section{Friedmann-Lema\^{\i}tre Equations}

\label{sec:Friedmanneq}

In this section we consider a cosmological model in the MTMVMG theory. Our starting point is to write down the metric ansatz for the dynamic and the background manifolds which are spatially flat,
\begin{eqnarray*}
    ds_{(d)}^2 &=& g_{\mu \nu} dx^\mu dx^\nu = -N^2(t) dt^2 + a^2(t) \delta_{i j} dx^i dx^j, \\
    ds_{(b)}^2 &=& f_{\mu \nu} dx^\mu dx^\nu = -M^2(t) dt^2 + \tilde{a}^2(t) \delta_{i j} dx^i dx^j.
\end{eqnarray*}
In the case at hand, the action \eqref{eq:aksimtmg1} simplifies to
\begin{widetext}
\begin{eqnarray}
    S_\text{MTMVMG-M} &=& -\int d^D x \, a^{D - 1} \Bigg\{ (D - 1) M_{\text{Pl}}^{D - 2} N H^2 - \frac{1}{2 N} \dot{\psi}^2 + N V(\psi) \nonumber \\
    && - \, W(\psi) \bigg( M \sum_{n = 0}^{D - 1} c_n A_nu^{D - n - 1} + N \sum_{n = 0}^{D - 1} c_{D - n} A_n u^n \bigg) \nonumber \\
    && + \, \frac{\lambda^2 M^2}{N} \Bigg[ \frac{D - 1}{D - 2} \frac{2}{M_{\text{Pl}}^{D - 2}} \left( \frac{W(\psi)}{4} \sum_{n = 1}^{D - 1} c_n B_n u^{D - n - 1} \right)^2 \nonumber \\
    && + \, \frac{1}{2} \left(\frac{dW}{d\psi} \sum_{n = 1}^{D - 1} c_n A_n u^{D - n - 1} \right)^2 \Bigg] - \lambda \Bigg[ \frac{\dot{\psi}M}{N} \frac{dW}{d\psi} \sum_{n = 1}^{D - 1} c_n A_n u^{D - n - 1} \nonumber \\
    && + \, (D - 1) W(\psi) M \left( H - u H_f \right) \sum_{n = 1}^{D - 1} c_n B_n u^{D - n - 1} \Bigg] \Bigg\} + S_\text{matter}.
\end{eqnarray}
Here, we have introduced the quantities
\begin{equation}
    H \equiv \frac{\dot{a}}{Na}, \qquad H_f \equiv \frac{\dot{\tilde{a}}}{M\tilde{a}}, \qquad u \equiv \frac{\tilde{a}}{a}, \label{eq:parameter}
\end{equation}
and parameters
\begin{eqnarray}
    A_n &\equiv& \frac{(D - 1)!}{(D - n - 1)! n!}, \\
    B_n &\equiv& \sum_{m = 1}^n (-1)^m \frac{(D - 1)!}{(D - n + m - 1)! (n - m)!}.
\end{eqnarray}
Then, the variation with respect to the lapse function $N(t)$ gives us the first Friedmann-Lema\^{\i}tre equation,
\begin{equation}
    \frac{1}{2} (D - 1) (D - 2) H^2 = \frac{1}{M_{\text{Pl}}^{D - 2}} (\rho_m + \rho_\text{MG} + \rho_\lambda), \label{eq:Friedman1}
\end{equation}
where
\begin{eqnarray}
    \rho_\text{MG} &\equiv& \frac{1}{2 N^2} \dot{\psi}^2 + V(\psi) + W(\psi) \sum_{n = 0}^{D - 1} c_{n + 1} A_n u^{D - n - 1}, \\
    \rho_\lambda &\equiv& \frac{\lambda \dot{\psi} M}{N^2} \frac{dW}{d\psi} \sum_{n = 1}^{D - 1} c_n A_n u^{D - n - 1} - \frac{D - 1}{D - 2} \frac{\lambda^2 M^2 W^2(\psi)}{2M_{\text{Pl}}^{D - 2} N^2} \left( \sum_{n = 1}^{D - 1} c_n B_n u^{D - n - 1} \right)^2 \nonumber \\
    && + \, \frac{(D - 1) \lambda H M W(\psi)}{N} \sum_{n = 1}^{D - 1} c_n B_n u^{D - n - 1} + \frac{\lambda^2 M^2}{2 N^2} \left(\frac{dW}{d\psi} \right)^2 \left( \sum_{n = 1}^{D - 1} c_n A_n u^{D - n - 1} \right)^2. \label{eq:rhograv}
\end{eqnarray}
The variation with respect to the scale factor $a(t)$ gives us the second Friedmann-Lema\^{\i}tre equation,
\begin{equation}
    (D - 2) \frac{\dot{H}}{N} + \frac{1}{2} (D - 1) (D - 2) H^2 = - \frac{1}{M_{\text{Pl}}^{D - 2}} (P_m + P_\text{MG} + P_\lambda), \label{eq:Friedman2}
\end{equation}
where
\begin{eqnarray}
    P_\text{MG} &\equiv& \frac{1}{2 N^2} \dot{\psi}^2 - V(\psi) - \frac{W(\psi)}{N} \sum_{n = 1}^{D - 1} \left(M c_n + N c_{n + 1} \right) n A_n u^{D - n - 1}, \\
    P_\lambda &\equiv& \frac{\lambda^2 M^2 W^2(\psi)} {2 M_{\text{Pl}}^{D - 2} N^2} \left( \sum_{n = 1}^{D - 1} \frac{D - 2 n - 1}{D - 2} c_n B_n u^{D - n - 1} \right) \left( \sum_{n = 1}^{D - 1} c_n B_n u^{D - n - 1} \right) \nonumber \\
    && - \, \frac{\lambda^2 M^2}{2 N^2} \left(\frac{dW}{d\psi} \right)^2 \left( \sum_{n = 1}^{D - 1} \frac{D - 2 n - 1}{D - 1} c_n A_n u^{D - n - 1} \right) \left( \sum_{n = 1}^{D - 1} c_n A_n u^{D - n - 1} \right) \nonumber \\
    && - \, \frac{\lambda M W(\psi) H_f}{N} \sum_{n = 1}^{D - 1} \Bigg( (D - n - 1) \frac{M}{N} + (n - 1) u \Bigg) c_n B_n u^{D - n - 1} \nonumber \\
    && - \, \left( \frac{\dot{\lambda} M N + \lambda (\dot{M} N - M \dot{N})}{N^3} \right) W(\psi) \sum_{n = 1}^{D - 1} c_n B_n u^{D - n - 1} \nonumber \\
    && + \, \frac{\lambda M \dot{\psi}}{N^2} \frac{dW}{d\psi} \sum_{n = 1}^{D - 1} c_n \left(\frac{n A_n}{D-1} + B_n \right) u^{D - n - 1}. \label{eq:Pgrav}
\end{eqnarray}
From the variation with respect to $\psi(t)$, we obtain the equation of motions,
\begin{eqnarray}
    && \frac{1}{N^2} \ddot{\psi} + \left( \frac{(D - 1) H}{N} - \frac{\dot{N}}{N^3} \right) \dot{\psi} + \frac{dV}{d\psi} + \frac{1}{N} \frac{dW}{d\psi} \Bigg\{ \sum_{n = 0}^{D - 1} \left(M c_n + N c_{n + 1} \right) A_n u^{D - n - 1} \nonumber \\
    && \quad + \, \frac{\lambda M}{N} \left( N H \sum_{n = 1}^{D - 1} n c_n A_n u^{D - n - 1} + M H_f \sum_{n = 1}^{D - 1} (D - n - 1) c_n A_n u^{D - n - 1} \right) \nonumber \\
    && \quad + \, \left( \frac{\dot{\lambda} M}{N} + \frac{\lambda (\dot{M} N - M \dot{N})}{N^2} \right) \sum_{n = 1}^{D - 1} c_n A_n u^{D - n - 1} - \lambda M \Bigg[ \frac{\lambda M}{N} \frac{d^2W}{d\psi^2} \left( \sum_{n = 1}^{D - 1} c_n A_n u^{D - n - 1} \right)^2 \nonumber \\
    && \quad - \, \frac{D - 1}{D - 2} \frac{\lambda M W(\psi)}{4 M_{\text{Pl}}^{D - 2} N} \left( \sum_{n = 1}^{D - 1} c_n B_n u^{D - n - 1} \right)^2 + (D - 1) \left( H -u H_f \right) \sum_{n = 1}^{D - 1} c_n B_n u^{D - n - 1} \Bigg] \Bigg\} = 0. \label{eq:dinamikamedanskalar}
\end{eqnarray}
Performing the variation with respect to $\lambda$ will give us
\begin{eqnarray}
    && W(\psi) \left( u H_f - H \right) \left( \sum_{n = 1}^{D - 1} c_n B_n u^{D - n - 1} \right) + \frac{\lambda}{D - 2} \frac{W^2(\psi) M}{M_{\text{Pl}}^{D - 2} N} \left( \sum_{n = 1}^{D - 1} c_n B_n u^{D - n - 1} \right)^2 \nonumber \\
    && \quad - \, \frac{\dot{\psi}}{D - 1} \frac{dW}{d\psi} \left( \sum_{n = 1}^{D - 1} c_n A_n u^{D - n - 1} \right) - \frac{\lambda}{D - 1} \frac{M}{N} \left( \frac{dW}{d\psi} \right)^2 \left( \sum_{n = 1}^{D - 1} c_n A_n u^{D - n - 1} \right)^2 = 0, \label{eq:parhubblelatarbelakang}
\end{eqnarray}
\end{widetext}
which relates the Hubble rate of the background spacetime with the Hubble rate of the dynamical spacetime. In the simple case where $\psi$ is trivial and $D = 4$, Eq.~\eqref{eq:parhubblelatarbelakang} corresponds to the branches of solutions discussed in \cite{defelice2016phenomenology}. In the case of $D = 5, 6$ with trivial $\psi$, one has to solve the qubic and the quartic polynomials, respectively, while for $D \ge 7$ the solutions of the polynomials are still unknown. For nontrivial $\psi$, it is still unknown whether such branches exist. These aspects will be considered elsewhere.

\section{Dynamical System Analysis}

\label{sec:Dynas}

Let us consider a special case where the couplings $W(\psi)$ and $V(\psi)$ have the form
\begin{eqnarray}
    W(\psi) &=& W_0 \exp \left( -\frac{\lambda_W \psi}{\sqrt{M_{\text{Pl}}^{D - 2}}} \right), \label{eq:VWpsi1} \\
    V(\psi) &=& V_0 \exp \left( -\frac{\lambda_V \psi}{\sqrt{M_{\text{Pl}}^{D - 2}}} \right), \label{eq:VWpsi2}
\end{eqnarray}
which has been considered in four dimensional cases \cite{leon2013cosmological, wu2013dynamical} where the constants $V_0, W_0 > 0$ and $\lambda_V, \lambda_W \geq 0$. Moreover, the form of the scalar potential $V(\psi)$ may provide inflationary expansion of the early universe model and has been well studied in the context of dynamical systems in Ref.~\cite{copeland1998exponential}. Note that for $\lambda_W = 0$ we have the MTMVMG theory with ordinary constant graviton mass \cite{defelice2016minimal,defelice2016phenomenology}.

For the rest of the paper we simply take a branch of solutions of Eq.~\eqref{eq:parhubblelatarbelakang} where $\lambda = 0$. Setting the lapse functions $N(t) = M(t) = 1$, we introduce the autonomous variables,
\begin{eqnarray}
    x_\rho &\equiv& \left( \frac{2 \rho_m}{M_{\text{Pl}}^{D - 2} \lambda_D^2 H^2} \right)^{1/2}, \\
    x_\psi &\equiv& \frac{\dot{\psi}}{\sqrt{M_{\text{Pl}}^{D - 2}} \lambda_D H}, \\
    x_V &\equiv& \left( \frac{2 V(\psi)}{M_{\text{Pl}}^{D - 2} \lambda_D^2 H^2} \right)^{1/2}, \\
    x_W &\equiv& \left( \frac{2 W(\psi)}{M_{\text{Pl}}^{D - 2} \lambda_D^2 H^2} \right)^{1/2},
\end{eqnarray}
with $\lambda^2_D \equiv (D - 1) (D - 2)$ such that the equations of motion \eqref{eq:Friedman2}, \eqref{eq:dinamikamedanskalar}, and \eqref{eq:parhubblelatarbelakang} can be written down in terms of the autonomous variables,
\begin{eqnarray}
    \frac{2}{D - 1} x_\psi' &=& (1 - w_m) x_\psi^3 - (1 + w_m) x_\psi x_V^2 \nonumber \\
    && - \, [f_2(u) + w_m f_1(u)] x_\psi x_W^2 - (1 - w_m) x_\psi \nonumber \\
    && + \, 2 \sqrt{\frac{D - 2}{D - 1}} [\lambda_V x_V^2 + \lambda_W f_1(u) x_W^2], \label{eq:xpsiprime} \\
    \frac{2}{D - 1} x_V' &=& (1 + w_m) x_V + (1 - w_m) x_\psi^2 x_V \nonumber \\
    && - \, (1 + w_m) x_V^3 - [f_2(u) + w_m f_1(u)] x_V x_W^2 \nonumber \\
    && - \, \sqrt{\frac{D - 2}{D - 1}} \lambda_V x_\psi x_V, \label{eq:xvprime} \\
    \frac{2}{D - 1} x_W' &=& (1 + w_m) x_W + (1 - w_m) x_\psi^2 x_W \nonumber \\
    && - \, (1 + w_m) x_V^2 x_W - [f_2(u) + w_m f_1(u)] x_W^3 \nonumber \\
    && - \, \sqrt{\frac{D - 2}{D - 1}} \lambda_W x_\psi x_W, \label{eq:xwprime}
\end{eqnarray}
where we have used the constraint coming from Eq.~\eqref{eq:Friedman1},
\begin{equation}
    x_\rho^2 + x_\psi^2 + x_V^2 + f_1(u) x_W^2 = 1. \label{eq:friedmanconstr}
\end{equation}
Note that here the prime symbol denotes the derivative with respect to $\ln{(a)}$. We also have defined
\begin{eqnarray}
    f_1(u) &\equiv& \sum_{n = 0}^{D - 1} c_{n + 1} A_n u^{D - n - 1}, \label{eq:f1f2-1} \\
    f_2(u) &\equiv& \sum_{n = 1}^{D - 1} \left( c_n + c_{n + 1} \right) n A_n u^{D - n - 1}, \label{eq:f1f2-2}
\end{eqnarray}
which are assumed to be bounded functions. As we have seen above, the scale factor $a(t)$ can be thought of as a parameter in this picture and the fiducial scale $\tilde{a}(t)$ is only a background in this setup. Therefore, we conclude that the quantity $u$ in Eq.~\eqref{eq:parameter} is not a dynamical variable; it might be either a function of time, $u(t)$, or a constant. The first case is called the \textit{normal branch}, while the latter is called the \textit{self-accelerating branch}. Both have been appeared in the context of four-dimensional MTMG \cite{defelice2016phenomenology}. Similar situation also occurs in the cosmological model of dRGT massive gravity (see, for example, Ref.~\cite{alatas2017parameter}).

We could also introduce some higher dimensional quantities which are analogous to the four-dimensional cases. First, the state equation parameter and the density parameter in this theory are given by
\begin{eqnarray}
    w_\text{MG} &\equiv& \frac{P_\text{MG}}{\rho_\text{MG}} = \frac{x_\psi^2 - x_V^2 - f_2(u) x_W^2}{x_\psi^2 + x_V^2 + f_1(u) x_W^2}, \\
    \Omega_\text{MG} &\equiv& \frac{2 \rho_\text{MG}}{M_{\text{Pl}}^{D - 2} \lambda_D^2 H^2} = x_\psi^2 + x_V^2 + f_1(u) x_W^2, \label{eq:mgsector}
\end{eqnarray}
respectively. Then, the decelerated parameter has the form
\begin{eqnarray}
    q &=& -1 - \frac{\dot{H}}{H^2} \nonumber \\
    &=& -1 + \frac{(D - 1)}{2} \Big[ 1 + w_m + (1 - w_m) x_\psi^2 \nonumber \\
    && - \, (1 + w_m) x_V^2 - (f_2(u) + w_m f_1(u)) x_W^2 \Big],
\end{eqnarray}
where for $q < 0$ we have an accelerated universe model. There are five critical points in the massless sector and three critical points in the massive sector, which are listed in the Table \ref{tab:massless} and \ref{tab:massive}, respectively, including their properties. Note that the quantities $\mathcal{A}_\pm(D, w_m, \lambda_V)$ and $\mathcal{B}_\pm(D, w_m, \lambda_W)$ mentioned in these tables are defined as
\begin{eqnarray}
    && \mathcal{A}_\pm(D, w_m, \lambda_V) \equiv \frac{1 + w_m}{3 - w_m} + \frac{(D - 2) \lambda_V^2}{(D - 2) (3 - w_m)} \nonumber \\
    && \qquad \qquad \times \, \Bigg\{ 1 \pm \Bigg[ 1 - \frac{2 (D - 1) (2 - w_m) (1 + w_m)}{(D - 2) \lambda_V^2} \nonumber \\
    && \qquad \qquad + \, \left( \frac{(D - 1)(1 + w_m)}{(D - 2) \lambda_V^2} \right)^2 \Bigg]^{1/2} \Bigg\} \label{eq:Apm} \\
    && \mathcal{B}_\pm(D, w_m, \lambda_W) \equiv \frac{1 + w_m}{3 - w_m} + \frac{(D - 2) \lambda_W^2}{(D - 1) (3 - w_m)} \nonumber \\
    && \qquad \qquad \times \, \Bigg\{ 1 \pm \Bigg[ 1 - \frac{2 (D - 1) (2 - w_m) (1 + w_m)}{(D - 2) \lambda_W^2} \nonumber \\
    && \qquad \qquad + \, \left( \frac{(D - 1) (1 + w_m)}{(D - 2) \lambda_W^2} \right)^2 \Bigg]^{1/2} \Bigg\}. \label{eq:Bpm}
\end{eqnarray}

Let us first discuss the massless sector in which there are five critical points, namely, CP$_1$,  CP$_2$, CP$_3$, CP$_4$ and CP$_5$ with trivial coupling $W_0 = 0$. The point CP$_1$ describes a matter-dominated era in which the stability behavior depends on the state parameter $w_m$. We find the late-time attractor identified by a stable node and the past-time attractor identified by an unstable node. There are three possible stabilities, namely stable node, unstable node, and saddle point. The universe is either acceleratedly expanded for $w_m < -{(D - 3)}/{(D - 1)}$ or unacceleratedly expanded for $w_m \geq -{(D - 3)}/{(D - 1)}$. The first case is however unphysical since it contains vacuum with $w_m = -1$, while in the latter case the matter of the universe could be dominated by dust ($w_m = 0$) or radiation ($w_m = \frac{1}{D - 1}$).

\begin{widetext}

\begin{table}[htbp]

\scriptsize

\begin{tabularx}{\columnwidth}{| >{\centering\arraybackslash}>{\hsize=0.5\hsize}X | >{\centering\arraybackslash}>{\hsize=0.75\hsize}X | >{\centering\arraybackslash}>{\hsize=1.55\hsize}X | >{\centering\arraybackslash}>{\hsize=0.3\hsize}X | >{\centering\arraybackslash}>{\hsize=1.0\hsize}X | >{\centering\arraybackslash}>{\hsize=0.65\hsize}X | >{\centering\arraybackslash}>{\hsize=1.1\hsize}X | >{\centering\arraybackslash}>{\hsize=1.0\hsize}X | >{\centering\arraybackslash}>{\hsize=2.15\hsize}X |}

\hline

Critical Points & $x_{\psi, c}$ & $x_{V, c}$ & $x_{W, c}$ & Existence & $w_{\text{MG}}$ & $\Omega_{\text{MG}}$ & $q$ & Stability \\

\hline
\hline

CP$_1$ & $0$ & $0$ & $0$ & always & undefined & $0$ & $\frac{D - 3 + w_m(D - 1)}{2}$ & stable node for $w_m < -1$, unstable node for $w_m > 1$, saddle point otherwise \\

\hline

CP$_2$ & $1$ & $0$ & $0$ & always & $1$ & $1$ & $D - 2$ & stable node for $w_m > 1$, $\lambda_V > 2 \sqrt{\frac{D - 1}{D - 2}}$, and $\lambda_W > 2 \sqrt{\frac{D - 1}{D - 2}}$, unstable node for $w_m < 1$, $\lambda_V < 2 \sqrt{\frac{D - 1}{D - 2}}$, and $\lambda_W < 2 \sqrt{\frac{D - 1}{D - 2}}$, saddle point otherwise \\

\hline

CP$_3$ & $-1$ & $0$ & $0$ & always & $1$ & $1$ & $D - 2$ & unstable node for $w_m < 1$, saddle point otherwise \\

\hline

CP$_4$ & $x_{\psi, c}$ & $0$ & $0$ & $w_m = 1$ and $0 < | x_{\psi, c} | < 1$ & $1$ & $x_{\psi, c}^2$ & $D - 2$ & unstable for $-1 < x_{\psi, c} < 0$, or $0 < x_{\psi, c} < 1$ with at least either $\lambda_V < \frac{2}{x_{\psi, c}} \sqrt{\frac{D - 1}{D - 2}}$ or $\lambda_W < \frac{2}{x_{\psi, c}} \sqrt{\frac{D - 1}{D - 2}}$, nonhyperbolic for $0 < x_{\psi, c} < 1$ \\

\hline

CP$_5$ & $\frac{\mathcal{A}_\pm}{\lambda_V} \sqrt{\frac{D - 1}{D - 2}}$ & $\frac{1}{\lambda_V} \sqrt{\frac{(D - 1) (2 - \mathcal{A}_\pm) \mathcal{A}_\pm}{2 (D - 2)}}$ & $0$ & $0 < \mathcal{A}_\pm < 2$ & $\frac{3 \mathcal{A}_\pm - 2}{\mathcal{A}_\pm + 2}$ & $\frac{(D - 1) (2 + \mathcal{A}_\pm) \mathcal{A}_\pm}{2 (D - 2) \lambda_V^2}$ & $\frac{(D - 1) \mathcal{A}_\pm}{2} - 1$ & see Fig.~\ref{fig:CP5} \\

\hline

\end{tabularx}

\caption{\label{tab:massless} The properties, the existence, the equation-of-state parameter $w_\text{MG}$, the density parameter $\Omega_\text{MG}$, the deceleration parameter $q$, and the stability conditions of the critical points of the autonomous system in the massless sector, $W(\psi) = 0$. Note that we have introduced the notation $\mathcal{A}_\pm$ in Eq.~\eqref{eq:Apm}.}

\end{table}

\begin{table}[htbp]

\scriptsize

\begin{tabularx}{\columnwidth}{| >{\centering\arraybackslash}>{\hsize=0.5\hsize}X | >{\centering\arraybackslash}>{\hsize=0.75\hsize}X | >{\centering\arraybackslash}>{\hsize=1.1\hsize}X | >{\centering\arraybackslash}>{\hsize=1.55\hsize}X | >{\centering\arraybackslash}>{\hsize=1.25\hsize}X | >{\centering\arraybackslash}>{\hsize=0.45\hsize}X | >{\centering\arraybackslash}>{\hsize=1.1\hsize}X | >{\centering\arraybackslash}>{\hsize=0.85\hsize}X | >{\centering\arraybackslash}>{\hsize=1.45\hsize}X |}

\hline

Critical Points & $x_{\psi, c}$ & $x_{V, c}$ & $x_{W, c}$ & Existence & $w_{\text{MG}}$ & $\Omega_{\text{MG}}$ & $q$ & Stability \\

\hline
\hline

CP$_6$ & $0$ & $\sqrt{\frac{\lambda_W}{\lambda_W - \lambda_V}}$ & $\sqrt{\frac{\lambda_V}{| f_1(u) | (\lambda_W - \lambda_V)}}$ & $\lambda_W > \lambda_V$ and $f_1(u) = f_2(u) < 0$ & $-1$ & $1$ & $-1$ & stable node for $w_m > -1$ and $\lambda_V \lambda_W < \frac{1}{4} \frac{D - 2}{D - 1}$, stable spiral for $w_m > -1$ and $\lambda_V \lambda_W > \frac{1}{4} \frac{D - 2}{D - 1}$, saddle point otherwise \\

\hline

CP$_7$ & $0$ & $\sqrt{1 - f_1(u) x_{W, c}^2}$ & $x_{W, c}$ & $0 \leq x_{W, c} \leq 1$, $\lambda_V = \lambda_W = 0$, and $f_1(u) = f_2(u) > 0$ & $-1$ & $1$ & $-1$ & unstable for $w_m < -1$, nonhyperbolic for $w_m > -1$ \\

\hline

CP$_8$ & $\frac{\mathcal{B}_\pm}{\lambda_W} \sqrt{\frac{D - 1}{D - 2}}$ & $0$ & $\frac{1}{\lambda_W} \sqrt{\frac{(D - 1) (2 - \mathcal{B}_\pm) \mathcal{B}_\pm}{2 (D - 2) f_1(u)}}$ & $0 < \mathcal{B}_\pm < 2$ & $\frac{3 \mathcal{B}_\pm - 2}{\mathcal{B}_\pm + 2}$ & $\frac{(D - 1) (2 + \mathcal{B}_\pm) \mathcal{B}_\pm}{2 (D - 2) \lambda_W^2}$ & $\frac{(D - 1) \mathcal{B}_\pm}{2} - 1$ & see Fig.~\ref{fig:CP8} \\

\hline

\end{tabularx}

\caption{\label{tab:massive} The properties, the existence, the equation-of-state parameter $w_\text{MG}$, the density parameter $\Omega_\text{MG}$, the deceleration parameter $q$, and the stability conditions of the critical points of the autonomous system in massive sector, $W(\psi)\neq 0$. Note that we have introduced the notation $\mathcal{B}_\pm$ in Eq.~\eqref{eq:Bpm}.}

\end{table}

\end{widetext}

At CP$_2$ we also find a late-time and past-time attractors according to the parameter values with also three possible stabilities, namely stable node, unstable node, and saddle point. At this point, gravitons can be thought of as nonphantom energy dominated by nonaccelerating expansion process of the universe ($\Omega_\text{MG} = 1$). The point CP$_3$ has similar properties as CP$_2$ in terms of its existence and the values of $w_m$, $\Omega_\text{MG}$, and $q$. At this point we do not have a late-time attractor (a stable node).

The point CP$_4$ exists on the interval $0 < | x_{\psi, c} | < 1$ which are unstable on $-1 < x_{\psi, c} < 0$, or on $0 < x_{\psi, c} < 1$ with at least either $\lambda_W < \frac{2}{x_{\psi, c}} \sqrt{\frac{D - 1}{D - 2}}$ or $\lambda_V < \frac{2}{x_{\psi, c}} \sqrt{\frac{D - 1}{D - 2}}$. It becomes nonhyperbolic on $0 < x_{\psi, c} < 1$ with $\lambda_V > \frac{2}{x_{\psi, c}} \sqrt{\frac{D - 1}{D - 2}}$ and $\lambda_W > \frac{2}{x_{\psi, c}} \sqrt{\frac{D - 1}{D - 2}}$. At this point, we have a class of universe with nonaccelerating expansion since $D \geq 4$. These universe are filled by the matter and the scalar field with constraint $x_{\rho, c}^2 + x_{\psi, c}^2 = 1$ and the equation of state parameters $w_m = w_\text{MG} = 1$.

\begin{figure}[t]
    \centering
    \includegraphics[width = 0.9 \columnwidth]{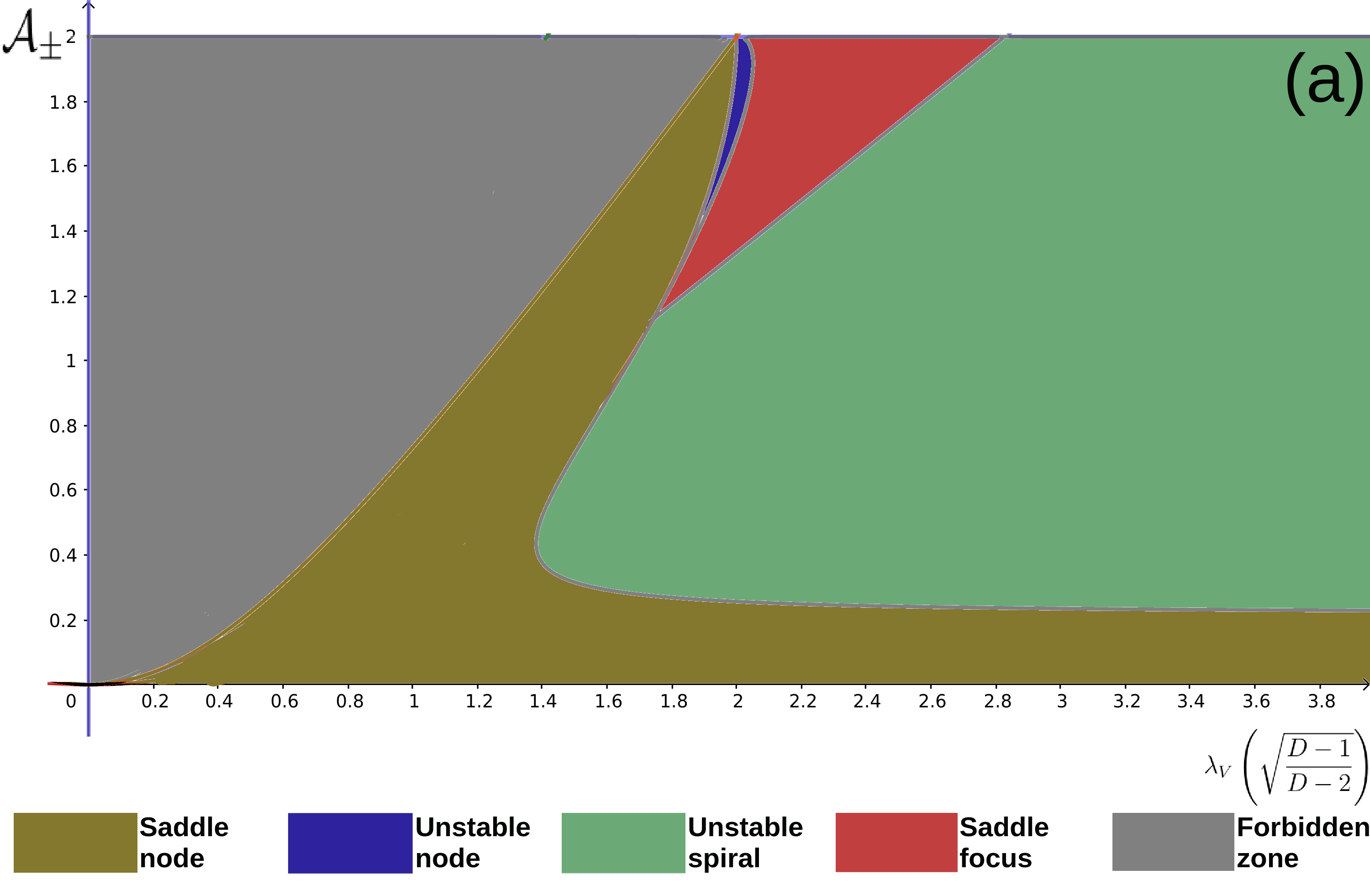}
    \includegraphics[width = 0.9 \columnwidth]{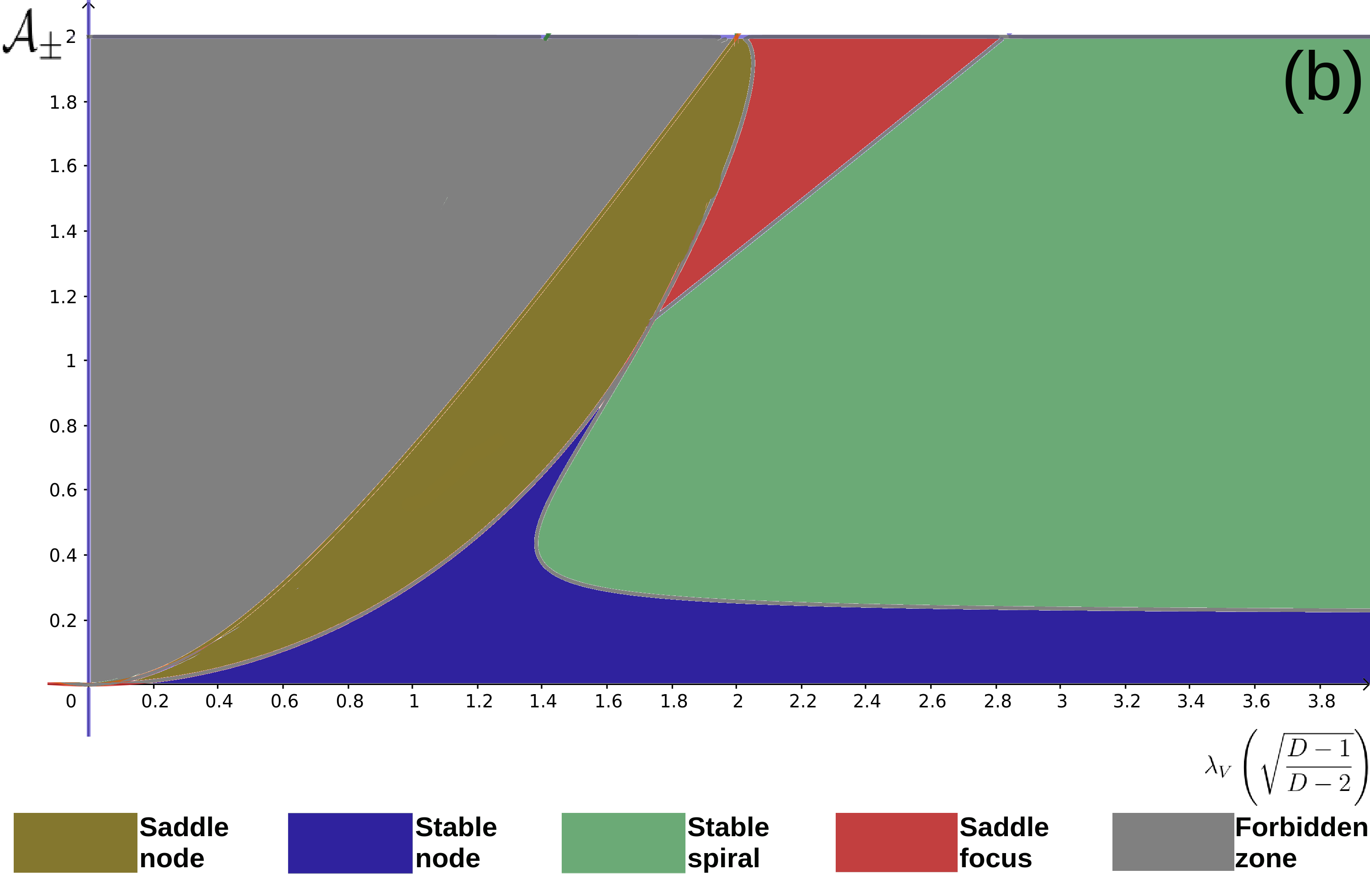}
    \caption{The stability conditions of the critical point CP$_5$ obtained by plotting $\mathcal{A}_\pm$ defined in Eq.~\eqref{eq:Apm} as a function of $\lambda_V$. The forbidden zone is the area with unphysical properties $\Omega_\text{MG} > 1$. This figure consist of two possibilities of the parameter values $\lambda_V$ and $\lambda_W$, where (a) for $\lambda_V > \lambda_W$, and (b) for $\lambda_V < \lambda_W$.}
    \label{fig:CP5}
\end{figure}

The point CP$_5$ exists on the interval $0 < \mathcal{A}_\pm < 2$ which could be either the late-time attractor, the past-time attractor, or saddle points, depending on the parameter $\lambda_V$ dan $\lambda_W$. For $\mathcal{A}_\pm(D, \lambda_V, \lambda_W) < \frac{2}{D - 1}$, we have a class of universe which acceleratedly expands and the scalar field plays a role as quintessencelike. The density parameter are on the interval $0 \leq \Omega_\text{MG} \leq 1$ such that $\lambda_V < \sqrt{\frac{(D - 1)(2 + \mathcal{A}_\pm) \mathcal{A}_\pm}{2 (D - 2)}}$ is not allowed since it will produce $\Omega_\text{MG} > 1$. Around $\lambda_V \approx \sqrt{\frac{(D - 1)(2 + \mathcal{A}_\pm) \mathcal{A}_\pm}{2(D - 2)}}$, the expansion of the universe is dominated by the scalar field ($\Omega_\text{MG} \approx 1$) which will be a good candidate for saddle power-law inflation model. This will be discussed in detail in Sec.~\ref{sec:cosmologicalcons}.

Next, we consider the massive sector which consists of the points CP$_6$, CP$_7$, and CP$_8$. The points CP$_6$ and CP$_7$ have vanishing kinetic part of the scalar field, which implies that the scalar becomes very massive. On the other hand, the kinetic part of the scalar is nonzero at CP$_8$.

The point CP$_6$ could be either the late-time attractor for $w_m > -1$ or saddle point for $w_m < -1$ where the graviton mass plays a role as the cosmological constant which dominates the accelerating expansion of universe. This cosmological constant has to be positive since the parameter $\lambda_W > \lambda_V$. These features give us a good candidate for the compatible description of the well-known observation result \cite{peebles2003cosmological}.

\begin{figure}[t]
    \centering
    \includegraphics[width = 0.9 \columnwidth]{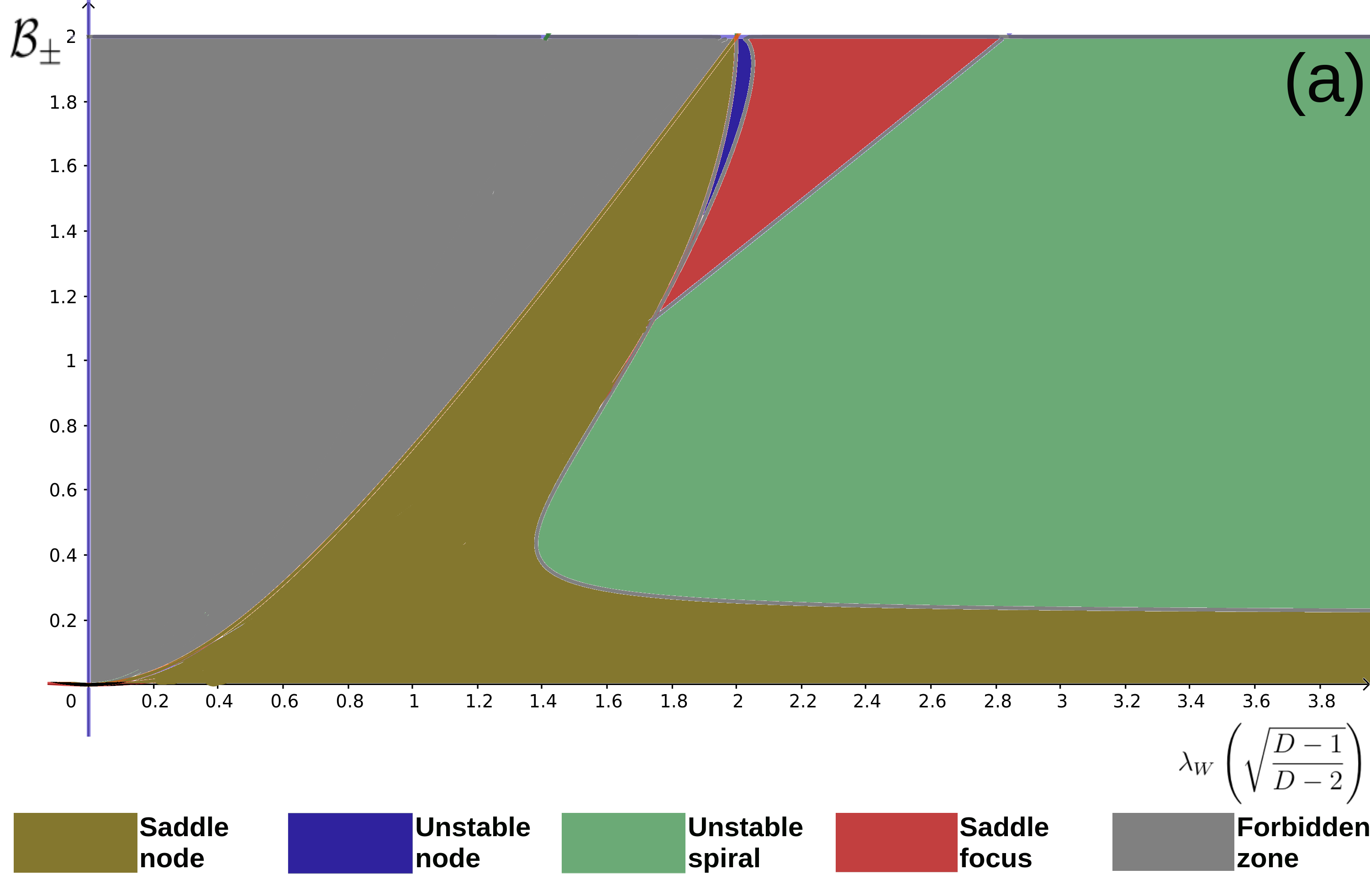}
    \includegraphics[width = 0.9 \columnwidth]{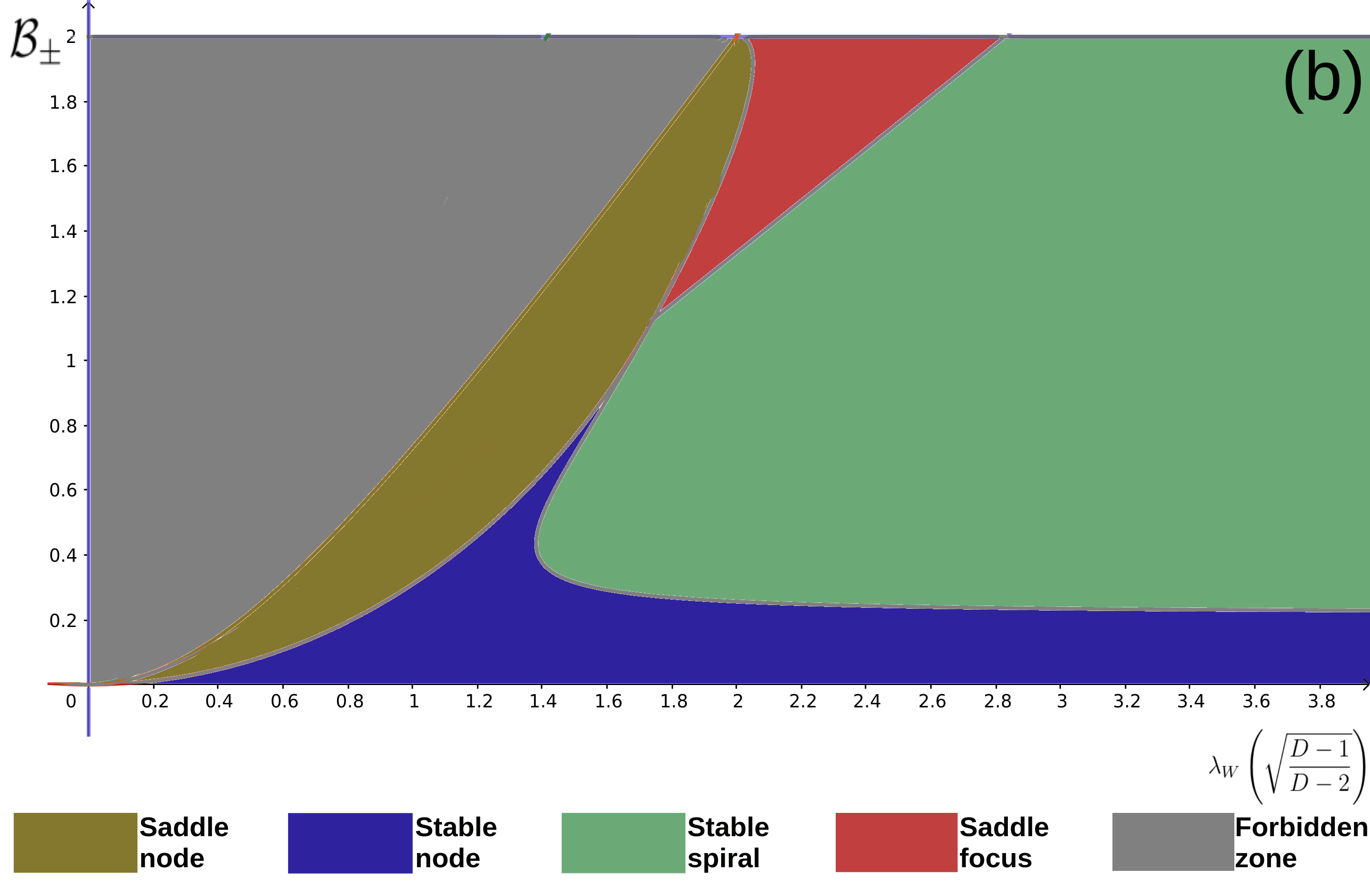}
    \caption{The stability conditions of the critical point CP$_8$ obtained by plotting $\mathcal{B}_\pm$ defined in Eq.~\eqref{eq:Bpm} as a function of $\lambda_W$. The forbidden zone is the area with unphysical properties $\Omega_\text{MG} > 1$. This figure consists of two possibilities of the parameter values $\lambda_V$ and $\lambda_W$, where (a) for $\lambda_V < \lambda_W$, and (b) for $\lambda_V>\lambda_W$.}
    \label{fig:CP8}
\end{figure}

The point CP$_7$ is defined on the circle $x_{V, c}^2 + f_1(u) x_{W, c}^2 = 1$ which may be either unstable for $w_m < -1$ or nonhyperbolic for $w_m > -1$. At this point, the parameter $\lambda_V = \lambda_W = 0$ showing that both the scalar potential and the graviton mass are constant. They behave like the cosmological constant which dominates the accelerating expansion of universe.

Finally, the point CP$_8$ which could be either the late-time attractor or the saddle point depending on the values of $\lambda_V$ and $\lambda_W$. For $\mathcal{B}_\pm(D, \lambda_V, \lambda_W) < \frac{2}{D - 1}$, the universe acceleratedly expands and the mass-varying massive graviton plays a role as quentessence dark energy. There exists a forbidden zone in CP$_8$; it has $\Omega_\text{MG} > 1$ for $\lambda_W < \sqrt{\frac{(D - 1)(2 + \mathcal{B}_\pm) \mathcal{B}_\pm}{2 (D - 2)}}$.

\section{Local-Global Existence of Solutions}

\label{sec:locglobex}

In this section we will prove the local-global existence and the uniqueness of the evolution equations \eqref{eq:xpsiprime}, \eqref{eq:xvprime}, and \eqref{eq:xwprime} with constraint  \eqref{eq:friedmanconstr} using Picard's iteration and the contraction mapping properties. We will first discuss the $f_1(u) > 0$ case, then continue with the $f_1(u) < 0$ case.

First of all, we  introduce the dynamical variables
\begin{equation}
    \bm{u} = \begin{pmatrix} x_\psi \\ x_V \\ x_W \end{pmatrix}, \label{dynvar}
\end{equation}
defined on an interval $I \equiv [s, s+ \epsilon]$ where $s \equiv \ln{a} \in \mathbb{R}$ and $\epsilon$ is a small positive constant. The functions $f_1(u)$ and $f_2(u)$ in Eqs.~\eqref{eq:f1f2-1} and \eqref{eq:f1f2-2}are bounded. In the first part of this section, we consider the case of $f_1(u) > 0$ such that from the constraint \eqref{eq:friedmanconstr} we could have
\begin{equation}
    \left.
    \begin{array}{cc}
        0 \leq | x_\psi | \leq 1 & \\
        0 \leq | x_V | \leq 1 & \\
        0 \leq | f_1(u) |^{1/2} | x_W | \leq 1 &
    \end{array}
    \right\rbrace \label{eq:syaratvardin}
\end{equation}
In other words, all of the quantities $(x_\psi, x_V, f_1(u) x_W)$ are defined on an open set $U \subset S^3$ where $S^3$ is the 3-sphere. It is important to notice that the critical point CP$_6$ is excluded in this setup.

All of the evolution equations \eqref{eq:xpsiprime}, \eqref{eq:xvprime}, and \eqref{eq:xwprime} can be simply rewritten into
\begin{equation}
    \frac{d\bm{u}}{ds} = \mathcal{J}(\bm{u}), \label{fungsiJ}
\end{equation}
with
\begin{widetext}
\begin{eqnarray}\label{JY}
    \mathcal{J}(\bm{u}) \equiv \frac{1}{2} (D-2) \left( \begin{array}{c}
        (1 - w_m) x_\psi^3 - (1 + w_m) x_\psi x_V^2 - (f_2(u) + w_m f_1(u)) x_\psi x_W^2 \\
        -(1 - w_m) x_\psi + 2 \sqrt{\frac{D - 2}{D - 1}} \left( \lambda_V x_V^2 + \lambda_W f_1(u) x_W^2 \right) \\ \\
        (1 + w_m) x_V + (1 - w_m) x_\psi^2 x_V - (1 + w_m) x_V^3 \\
        -(f_2(u) + w_m f_1(u)) x_V x_W^2 - \sqrt{\frac{D - 2}{D - 1}} \lambda_V x_\psi x_V \\ \\
        (1 + w_m) x_W + (1 - w_m) x_\psi^2 x_W - (1 + w_m) x_V^2 x_W \\
        -(f_2(u) + w_m f_1(u)) x_W^3 - \sqrt{\frac{D - 2}{D - 1}} \lambda_W x_\psi x_W
    \end{array} \right).
\end{eqnarray}
\end{widetext}

\begin{lemma} \label{opJY}
    The operator $ \mathcal{J}(\bm{u})$ in Eq.~\eqref{fungsiJ} is locally Lipschitz with respect to $\bm{u}$.
\end{lemma}

\begin{proof}
We have the following estimate
\begin{eqnarray}
    && | \mathcal{J} |_U \leq \frac{1}{2} (D-2) \Bigg[ | 1 - w_m | | x_\psi |^3 + | 1 + w_m | | x_\psi | | x_V |^2 \nonumber \\
    && \qquad + \, | f_2(u) + w_m f_1(u) | | x_\psi | | x_W |^2 + | 1 - w_m | | x_\psi | \nonumber \\
    && \qquad + \, 2 \sqrt{\frac{D - 2}{D - 1}} \left( \lambda_V | x_V |^2 + \lambda_W | f_1(u) | | x_W |^2 \right) \nonumber \\
    && \qquad + \, | 1 + w_m | | x_V | + | 1 - w_m | | x_\psi |^2 | x_V | + | 1 + w_m | | x_V |^3 \nonumber \\
    && \qquad + \, | f_2(u) + w_m f_1(u) | | x_V | | x_W |^2 + \sqrt{\frac{D - 2}{D - 1}} \lambda_V | x_\psi | | x_V | \nonumber \\
    && \qquad + \, | 1 + w_m | | x_W | + | 1 - w_m | | x_\psi |^2 | x_W | \nonumber \\
    && \qquad + \, | 1 + w_m | | x_V |^2 | x_W | + | f_2(u) + w_m f_1(u) | | x_W |^3 \nonumber \\
    && \qquad + \, \sqrt{\frac{D - 2}{D - 1}} \lambda_W | x_\psi | | x_W | \Bigg]. \label{estJY}
\end{eqnarray}
Then, using Eq.~\eqref{eq:syaratvardin}, we can show that $| \mathcal{J}(\bm{u}, x) |_U$ is indeed bounded on $U$.

Moreover, for $\bm{u}, \hat{\bm{u}} \in U$ we have
\begin{eqnarray}\label{estJYLps}
    && | \mathcal{J}(\bm{u}) - \mathcal{J}(\hat{\bm{u}}) |_U \leq \frac{1}{2} (D-2) \Bigg[ | 1 - w_m | | x_\psi^3 - \hat{x}_\psi^3 | \nonumber \\
    && \qquad + \, | 1 + w_m | | x_\psi x^2_V - \hat{x}_\psi \hat{x}^2_V | + | 1 - w_m | | x_\psi - \hat{x}_\psi | \nonumber \\
    && \qquad + \, | f_2(u) + w_m f_1(u) | | x_\psi x^2_W - \hat{x}_\psi \hat{x}^2_W | \nonumber \\
    && \qquad + \, 2 \sqrt{\frac{D - 2}{D - 1}} \left( \lambda_V | x^2_V - \hat{x}^2_V | + \lambda_W | f_1(u) | | x^2_W - \hat{x}^2_W | \right) \nonumber \\
    && \qquad + \, | 1 + w_m | | x_V - \hat{x}_V | + | 1 - w_m | | x^2_\psi x_V - \hat{x}^2_\psi \hat{x}_V | \nonumber \\
    && \qquad + \, | 1 + w_m | | x^3_V - \hat{x}^3_V | \nonumber \\
    && \qquad + \, | f_2(u) + w_m f_1(u) | | x_V x^2_W - \hat{x}_V \hat{x}^2_W | \nonumber \\
    && \qquad + \, \sqrt{\frac{D - 2}{D - 1}} \lambda_V | x_\psi x_V - \hat{x}_\psi \hat{x}_V | \nonumber \\
    && \qquad + \, | 1 + w_m | | x_W - \hat{x}_W | + | 1 - w_m | | x_\psi^2 x_W - \hat{x}_\psi^2 \hat{x}_W | \nonumber \\
    && \qquad + \, | 1 + w_m | | x_V^2 x_W - \hat{x}_V^2 \hat{x}_W | \nonumber \\
    && \qquad + \, | f_2(u) + w_m f_1(u) | | x^3_W - \hat{x}^3_W | \nonumber \\
    && \qquad + \, \sqrt{\frac{D - 2}{D - 1}} \lambda_W | x_\psi x_W - \hat{x}_\psi \hat{x}_W | \Bigg].
\end{eqnarray}
After some computations, we obtain
\begin{equation}
    \left| \mathcal{J}(\bm{u}) - \mathcal{J}(\hat{\bm{u}}) \right|_U \le C_{\mathcal{J}}(| \bm{u} |, | \hat{\bm{u}} |) | \bm{u} - \hat{\bm{u}} |, \label{localLipshitzcon}
\end{equation}
showing that $\mathcal{J}$ is locally Lipshitz with respect to $\bm{u}$.
\end{proof}

Next, we rewrite Eq.~\eqref{fungsiJ} into the integral form
\begin{equation}
    \bm{u}(s) = \bm{u}(s_0) + \int_{s_0}^s \, \mathcal{J} \left( \bm{u}(\hat{s}) \right) d\hat{s}. \label{IntegralEquation}
\end{equation}
We define a Banach space
\begin{equation}
    X \equiv \{ \bm{u} \in C(I, \mathbb{R}^2) : \, \bm{u}(x_0) = \bm{u}_{0}, \, \sup_{x \in I}{| \bm{u}(x) |} \leq L_0 \},
\end{equation}
endowed with the norm
\begin{equation}
    | \bm{u} |_{X} = \sup_{x \in I}{|\bm{u}(x)|},
\end{equation}
where $L_0 > 0$. Introducing an operator $\mathcal{K}$
\begin{equation}
    \mathcal{K}(\bm{u}(x)) = \bm{u}_{0} + \int_{x_0}^x \mathcal{J} \left( \bm{u}(s), s \right) ds, \label{OpKdefinition}
\end{equation}
and using Lemma \ref{opJY}, we have the following result \cite{akbar2015existence}:
\begin{lemma} \label{uniqueness}
    Let $\mathcal{K}$ be an operator defined in Eq.~\eqref{OpKdefinition}. Suppose there exists a constant $\varepsilon > 0$ such that $\mathcal{K}$ is a mapping from $X$ to itself and $\mathcal{K}$ is a contraction mapping on $I = [x, x + \varepsilon]$ with
    \begin{equation}
        \varepsilon \leq \min \left( \frac{1}{C_{L_0}}, \frac{1}{C_{L_0} L_0 + \| \mathcal{J}(x) \|} \right).
    \end{equation}
    Then, the operator $\mathcal{K}$ is a contraction mapping on $X$.
\end{lemma}
\noindent The above lemma shows that there exists a unique fixed point of Eq.~\eqref{OpKdefinition} ensuring a unique local solution of the differential equation \eqref{fungsiJ}. We can further construct a maximal solution by repeating the above arguments of the local existence with the initial condition $\bm{u}(x - x_n)$ for some $x_0 < x_n < x$ and using the uniqueness condition to glue the solutions.

We can now show the existence of global solutions of Eq.~\eqref{fungsiJ}. Let us consider the integral form \eqref{IntegralEquation} such that
\begin{equation}
    | \bm{u}(s) | \le | \bm{u}(s_0) | + \int_{s_0}^s | \mathcal{J} \left( \bm{u}(\hat{s}) \right) | d\hat{s}. \label{IntegralEquation1}
\end{equation}
We first consider the self-accelerating branch where the parameter $u$ in Eq.~\eqref{eq:parameter} is constant. Using Eqs.~\eqref{eq:syaratvardin} and \eqref{estJY}, we get
\begin{eqnarray}
    && | \bm{u}(t) | \leq | \bm{u}(t_0) | + \frac{1}{2} (D - 2) \Bigg[ 3 | 1 - w_m | + 3 | 1 + w_m | \nonumber \\
    && \qquad + \, | f_2(u) + w_m | + \sqrt{\frac{D - 2}{D - 1}} \left( 3 \lambda_V + 2 \lambda_W \right) \nonumber \\
    && \qquad + \, 2 \left| \frac{f_2(u)}{f_1(u)} + w_m \right| + 2 \frac{| 1 + w_m |}{| f_1(u) |^{1/2}} + \frac{| 1 - w_m |}{| f_1(u) |^{1/2}} \nonumber \\
    && \qquad + \, \frac{| f_2(u) + w_m f_1(u) |}{| f_1(u) |^{3/2}} \nonumber \\
    && \qquad + \, \sqrt{\frac{D - 2}{D - 1}} \frac{\lambda_W}{| f_1(u) |^{1/2}} \Bigg] \ln{\left( \frac{a(t)}{a(t_0)} \right)}. \label{eq:solIntegralEquation}
\end{eqnarray}

The second part is to consider the $f_1(u) < 0$ case. From the constraint \eqref{eq:friedmanconstr} we could have
\begin{equation}
    \left.
    \begin{array}{rcl}
        x_\psi & = & \cos{\alpha} \\
        x_V & = & \sin{\alpha} \cosh{\beta} \\
        |f_1(u)|^{1/2} x_W & = & \sin{\alpha} \sinh{\beta}
    \end{array}
    \right\rbrace \label{eq:syaratvardin1}
\end{equation}
where $\alpha \equiv \alpha(s)$ and $\beta \equiv \beta(s)$. In this case, Lemma \ref{opJY} and Lemma \ref{uniqueness} still hold, but we need to modify the estimate to show the global existence. In the case at hand, using Eq.~\eqref{eq:syaratvardin1}, the estimate \eqref{IntegralEquation1} becomes
\begin{eqnarray}
    && | \bm{u}(t) | \leq | \bm{u}(t_0) | + \frac{1}{2} (D - 2) \Bigg\{ 2 | 1 - w_m | \ln{\left( \frac{a(t)}{a (t_0)} \right)} \nonumber \\
    && \qquad + \, \Bigg[ | 1 + w_m | + | 1 - w_m | + \frac{| 1 + w_m |}{| f_1(u) |^{1/2}} + \frac{| 1 - w_m |}{| f_1(u) |^{1/2}} \nonumber \\
    && \qquad + \, \sqrt{\frac{D - 2}{D - 1}} \left( \lambda_V + \frac{\lambda_W}{| f_1(u) |^{1/2}} \right) \Bigg] \int_{s_0}^s \cosh{\beta} \, d\hat{s} \nonumber \\
    && \qquad + \, \Bigg[ | 1 + w_m | + \left| \frac{f_2(u)}{f_1(u)} + w_m \right| \nonumber \\
    && \qquad + \, 2 \sqrt{\frac{D - 2}{D - 1}} \left( \lambda_V + \lambda_W \right) \Bigg] \int_{s_0}^s \cosh^2{\beta} \, d\hat{s} \nonumber \\
    && \qquad + \, \Bigg[ 2 | 1 + w_m | + \left| \frac{f_2(u)}{f_1(u)} + w_m \right| \nonumber \\
    && \qquad + \, \frac{| f_2(u) + w_m f_1(u) |}{| f_1(u) |^{3/2}} \Bigg] \int_{s_0}^s \cosh^3{\beta} \, d\hat{s} \Bigg\}. \label{eq:solIntegralEquation1}
\end{eqnarray}
For the normal branch with $u(t)$, we employ similar procedure as above and use the assumptions that $f_1(u)$ and $f_2(u)$ are bounded. Then, we obtain the slightly modified forms of Eqs.~\eqref{eq:solIntegralEquation} and \eqref{eq:solIntegralEquation1}.

Thus, we have proven
\begin{theorem} \label{thmlocglob}
    There exists a global solution of the evolution equations \eqref{eq:xpsiprime}, \eqref{eq:xvprime}, and \eqref{eq:xwprime} with constraint  \eqref{eq:friedmanconstr}.
\end{theorem}

\section{Cosmological Models}

\label{sec:cosmologicalcons}

In this section we will discuss some possible cosmological models of the theory. To simplify the computation, we particularly choose the self-accelerating branch where the parameter $u$ in Eq.~\eqref{eq:parameter} is constant. In the case of the exponential form of the potentials \eqref{eq:VWpsi1} and \eqref{eq:VWpsi2}, we may have an inflation era in which it can be described by the well-known power law inflation \cite{lucchin1985power} where $a(t) \varpropto t^{1/\epsilon}$, with the slow-roll parameter $\epsilon = {| \dot{H} |}/{H^2} < 1$.

If the scalar $\psi$ plays a role as the inflaton field in the early epoch, then the critical points CP$_5$ and CP$_8$ are the good candidates to describe that era, with the slow-roll parameter given by
\begin{eqnarray}
    \epsilon &=& \frac{(D - 2) \lambda_V^2}{4} < 1, \label{eq:epsiloncp5} \\
    \epsilon &=& \frac{(D - 2) \lambda_W^2}{4} \left( 1 + \frac{\sum_{n = 0}^{D - 1} c_n A_n u^{D - n - 1}}{\sum_{n = 0}^{D - 1} c_{n + 1} A_n u^{D - n - 1}} \right) < 1, \nonumber \\ \label{eq:epsiloncp8}
\end{eqnarray}
for CP$_5$ and CP$_8$, respectively. In the CP$_5$ case, which is in the massless sector, we can use Eqs.~\eqref{eq:mgsector} and \eqref{eq:epsiloncp5} to get
\begin{equation}
    w_\text{MG} = \frac{D - 2}{D - 1} \frac{\lambda_V^2}{2} - 1,
\end{equation}
which is negative for $D > 3$. Hence, from Table \ref{tab:massless}, we have
\begin{equation}
    \mathcal{A}_\pm = \frac{2 (D - 2) \lambda_V^2}{8 (D - 1) - (D - 2) \lambda_V^2} < \frac{1}{D - 1}.
\end{equation}
Similarly, in the CP$_8$ case, which is in the massive sector, assuming that $u$ is constant, we can use Eqs.~\eqref{eq:mgsector} and \eqref{eq:epsiloncp8} to get
\begin{equation}
    w_\text{MG} = \frac{D - 2}{D - 1} \frac{\lambda_W^2}{2} \left( 1 + \frac{\sum_{n = 0}^{D - 1} c_n A_n u^{D - n - 1}}{\sum_{n = 0}^{D - 1} c_{n + 1} A_n u^{D - n - 1}} \right) - 1,
\end{equation}
which is negative for $D > 3$. Hence, from Table \ref{tab:massive}, we obtain
\begin{widetext}
\begin{equation}
    \mathcal{B}_\pm = \frac{2 (D - 2) \lambda_W^2 \sum_{n = 0}^{D - 1} (c_n + c_{n + 1}) A_n u^{D - n - 1}}{8 (D - 1) \sum_{n = 0}^{D - 1} c_{n + 1} A_n u^{D - n - 1} - (D - 2) \lambda_W^2 \sum_{n = 0}^{D - 1} (c_n + c_{n + 1}) A_n u^{D - n - 1}} < \frac{1}{D - 1}.
\end{equation}
\end{widetext}

The scalar $\psi$ in the inflation era has the form
\begin{equation}
    \psi(t) \propto \frac{\sqrt{M_{\text{Pl}}^{D - 2}}}{\lambda_\alpha} \ln{\left[\frac{\lambda_\alpha^2 V_0 \epsilon t^2}{2 M_{\text{Pl}}^{D - 2} (D - 1 - \epsilon)} \right]}, \label{eq:psiinfl}
\end{equation}
where $\lambda_\alpha = \lambda_V$ ($\lambda_W$) for the case of CP$_5$ (CP$_8$), respectively.

In order to have a nucleosynthesis era, we have to assume $\epsilon \ll 1$ such that the reheating phase begins at $t_i$ and ends at $t_f$. Therefore, we have the scale factor
\begin{equation}
    \frac{a(t_f)}{a(t_i)} = \left( \frac{t_f}{t_i} \right)^{1/\epsilon} = \exp{\left[ \frac{\lambda_\alpha (\psi_f - \psi_i)}{2 \epsilon \sqrt{M_{\text{Pl}}^{D - 2}}} \right]}. \label{eq:afaiinfl}
\end{equation}
In addition, the scalar potential in $\psi_i = \psi(t_i)$ drops to the scalar potential in $\psi_f = \psi(t_f)$ implying that the argument of the exponential is positive and the scale factor increases during this period. At the end of the inflation (or in the beginning of the reheating at $t = t_i$), the scalar field $\psi$ still dominates the universe with initial density parameter $\Omega_\text{MG, i} \simeq 1$. Then, the phase transition of $\psi$ occurs such that $\psi$ decays to matter fields in the time interval $t_i < t < t_f$. The decay process of the scalar $\psi$ becoming a matter field occurs relativistically with $P_{m^\ast} = {\rho_{m^\ast}}/{(D - 1)}$. From Eq.~\eqref{eq:afaiinfl} and the conservation of energy of relativistic matter (radiation), we obtain the relation $\rho_{\text{MG}, f}/{\rho_{\text{MG}, i}} = ({a_i}/{a_f})^{2 \epsilon}$ and $\rho_{M, f}/\rho_{M, i} = ({a_i}/{a_f})^D$, respectively. Since the inflaton density parameter $\Omega_\text{MG} = 2 \rho_\text{MG}/{(M_{\text{Pl}}^{D - 2} \lambda_D^2 H^2)}$ and the matter field density, we find that $\Omega_\text{MG}$ decreases according to
\begin{equation}
    \Omega_\text{MG} \simeq \left(\frac{1 \ \text{MeV}^4}{\rho_M} \right)^{2 \epsilon/D}, \qquad (t_i \le t \le t_f),
\end{equation}
as $\rho_M$ increases. Note that here we have assumed that $\epsilon = {|\dot{H}|}/{H^2} \ll 1$ in the interval. At the end, we have $\rho_M \sim 10^4-10^8 \text{ MeV}^4$ at $t = t_f$.

On the other hand, we may also apply the theory to the case of late-time era. Here, we have at least two interesting possible scenarios. The first is that the scalar field $\psi$ becomes frozen, after the reheating era, at fixed value $\psi_\infty$ with $w_\text{MG} = -1$ at CP$_6$. The value of $W(\psi_\infty)$ then becomes the graviton mass in the present time, which can be determined using observational constraints. Hence, the accelerating expansion of the universe is due to the constant mass term, as in the dRGT theory. The second possible scenario is that the scalar field $\psi$ could play a role as dynamical quintessensial dark energy, either from CP$_5$ in the massless sector or from CP$_8$ in the massive sector. In the massless sector, the accelerating expansion of the universe is due solely to the standard quintessence paradigm. However, in the massive sector, it is due to the nontrivial interplay between quintessence and massive gravity in the massive sector.

\section{Conclusions}

\label{sec:conclusions}

We have constructed higher-dimensional MTMVMG with nonzero scalar potential, where the graviton mass is varied with respect to the real scalar field $\psi$. Our construction can be summarized as follows. First, we write the MVMG action for higher dimensions using the vielbein formulation in the ADM formalism. We also adopt the vielbein potential in Ref.~\cite{hinterbichler2012interacting} and couple it to the masslike scalar potential $W(\psi)$. Inserting the ansatz metric \eqref{eq:basis} and \eqref{eq:dualbasis} into the MVMG action, we then derive the precursor action \eqref{eq:Spre}. By imposing the set of $D$-constraints \eqref{Dconstraints1} and \eqref{Dconstraints2} to the precursor action \eqref{eq:Spre}, we obtain a theory in which the graviton has $D (D - 3)/2$ degrees of freedom for general $D$ spacetime dimensions, without scalar and vector modes of the St\"uckelberg field. The resulting theory admits the Lorentz violation since we have used the ADM vielbein \eqref{eq:basis} and \eqref{eq:dualbasis}. Still, the $O(D - 1)$ symmetry and the spatial diffeomorphism are preserved in the MTMVMG action \eqref{eq:aksimtmg1}. This theory generalizes the four-dimensional MTMG in Refs.~\cite{defelice2016minimal,defelice2016phenomenology} and the four-dimensional MVMG in Ref.~\cite{huang2012mass}.

Then, we derive the Friedmann-Lema\^itre equations \eqref{eq:Friedman1} and \eqref{eq:Friedman2} for the case of spatially flat spacetimes. To proceed, inspired by Refs.~\cite{copeland1998exponential,leon2013cosmological,wu2013dynamical}, we take both the scalar potential and the graviton mass couplings to have exponential forms \eqref{eq:VWpsi1} and \eqref{eq:VWpsi2}, such that Eqs.~\eqref{eq:Friedman1} and \eqref{eq:Friedman2} can be written into a set of autonomous equations \eqref{eq:xwprime} with constraint \eqref{eq:friedmanconstr}. By performing dynamical system analysis, we find that in this theory there exists five critical points in the massless sector, namely CP$_{1 - 5}$, whereas in the massive sector we have three critical points, namely CP$_{6 - 8}$. We also discuss their stability, their existence, and their cosmological aspects related to the state equation parameter $w_\text{MG}$, the density parameter $\Omega_\text{MG}$, and the decelerated parameter $q$. Among them, we may have some critical points that are suitable to explain either inflation phenomenon or the accelerated universe in the late-time era.

We also have established the local-global existence and the uniqueness of the evolution equations \eqref{eq:xpsiprime}, \eqref{eq:xvprime}, and \eqref{eq:xwprime} with constraint \eqref{eq:friedmanconstr} using Picard's iteration and the contraction mapping properties, assuming that the functions $f_1(u)$ and $f_2(u)$ are bounded. The discussion is then divided into two parts: the $f_1(u) > 0$ case and the $f_1(u) < 0$ case. Note that our results apply to all branches, namely the self-accelerating branch and the normal branch.

Finally, we have particularly discussed some possible cosmological models of the MTMVMG in the self-accelerating branch. Since both the scalar potential and the graviton mass couplings have exponential forms \eqref{eq:VWpsi1} and \eqref{eq:VWpsi2}, the theory has a good description of the inflation era in the early universe using the power-law inflation \cite{lucchin1985power} in which the scale factor $a(t) \varpropto t^{1/\epsilon}$ with the slow-roll parameter $\epsilon = {| \dot{H} |}/{H^2} < 1$. This era can be described either by the critical point CP$_5$ or CP$_8$. In other words, our theory can describe the inflationary era using both the massless and the massive sectors. Also, we have shown that the MTMVMG could accommodate the reheating mechanism in this framework, again for both massless and massive sectors. Perturbative approach needs to be applied, for example, to study the behavior of primordial gravitational waves based on MTMVMG, as in the case of four-dimensional MTMG \cite{fujita2019blue, fujita2020primordial}. The detailed construction and the phenomenological predictions are left for subsequent works. On the other hand, we have at least two interesting possible scenarios for the late times. The first scenario is that the dark energy in the present time is due to the graviton mass which depends on the scalar field $\psi_\infty$ that becomes frozen after the reheating era. The second scenario is that the scalar field $\psi$ plays role as dynamical quintessential dark energy. Therefore, contrary to the massless sector where the accelerating expansion is due to the standard quintessence paradigm, in the massive sector it is due to the nontrivial interplay between quintessence and massive gravity.

\section*{Acknowledgments}

The work in this paper is supported by P2MI FMIPA ITB 2021 and Riset ITB 2021. The work of H.~A. is partially funded by the World Class Research (WCR) grant from RistekBRIN-IPB 2021.

\appendix

\section{Mass Term}

\label{sec:appendixA}

In this Appendix we are going to evaluate potential element in Eq.~\eqref{eq:gfpot} by taking the ADM vielbein in Eqs.~\eqref{eq:basis} and \eqref{eq:dualbasis}. For $n = 0$ we obtain
\begin{eqnarray}
    && \frac{1}{D!} \hat{\epsilon}_{A_1 A_2 \cdots A_D} \tensor{\bm{\tilde{E}}}{^{A_1}} \wedge \tensor{\bm{\tilde{E}}}{^{A_2}} \wedge \tensor{\bm{\tilde{E}}}{^{A_{3}}} \wedge \dots \wedge \tensor{\bm{\tilde{E}}}{^{A_D}} \nonumber \\
    && \quad = \frac{d^D x}{D!} \hat{\epsilon}_{A_1 A_2 \cdots A_D} \hat{\epsilon}^{\mu_1 \mu_2 \cdots \mu_D} \tensor{\tilde{E}}{^{A_1}_{\mu_1}} \tensor{\tilde{E}}{^{A_2}_{\mu_2}} \tensor{\tilde{E}}{^{A_{3}}_{\mu_3}} \cdots \tensor{\tilde{E}}{^{A_D}_{\mu_D}} \nonumber \\
    && \quad = \frac{d^D x}{D!} \hat{\epsilon}_{A_1 A_2 \cdots A_D} \hat{\epsilon}^{B_1 B_2 \cdots B_D} \tensor{\delta}{^{A_1}_{B_1}} \tensor{\delta}{^{A_2}_{B_2}} \cdots \tensor{\delta}{^{A_D}_{B_D}} \det{(\tilde{E})} \nonumber \\
    && \quad = d^D x \, M \det{(\tilde{e})}.
\end{eqnarray}
Note that we have used the relations $dx^{\mu_1} \wedge dx^{\mu_2} \wedge \cdots \wedge dx^{\mu_D} = \hat{\epsilon}^{\mu_1 \mu_2 \cdots \mu_D} d^D x$ in first step and $\hat{\epsilon}^{\mu_1 \mu_2 \cdots \mu_D} = \hat{\epsilon}^{B_1 B_2 \cdots B_D} \tensor{\tilde{E}}{^{\mu_1}_{B_1}} \tensor{\tilde{E}}{^{\mu_2}_{B_2}} \cdots \tensor{\tilde{E}}{^{\mu_D}_{B_D}}$ in the second step. Following the same way as above for other $n$, we have
\begin{widetext}
\begin{eqnarray*}
    \hat{\epsilon}_{A_1 A_2 \cdots A_D} \tensor{\bm{E}}{^{A_1}} \wedge \tensor{\bm{\tilde{E}}}{^{A_2}} \wedge \tensor{\bm{\tilde{E}}}{^{A_3}} \wedge \dots \wedge \tensor{\bm{\tilde{E}}}{^{A_D}} &=& (D - 1)! \Big( M \det{(\tilde{e})} [\tilde{e}^{-1} e] + N \det{(\tilde{e})} \Big), \\
    \hat{\epsilon}_{A_1 A_2 \cdots A_D} \tensor{\bm{E}}{^{A_1}} \wedge \tensor{\bm{E}}{^{A_2}} \wedge \tensor{\bm{\tilde{E}}}{^{A_3}} \wedge \dots \wedge \tensor{\bm{\tilde{E}}}{^{A_D}} &=& (D - 2)! \Big( M \det{(\tilde{e})} ([\tilde{e}^{-1} e]^2 - [(\tilde{e}^{-1} e)^2]) + 2 N \det{(\tilde{e})} [\tilde{e}^{-1} e] \Big), \\
    \hat{\epsilon}_{A_1 A_2 \cdots A_D} \tensor{\bm{E}}{^{A_1}} \wedge \tensor{\bm{E}}{^{A_2}} \wedge \tensor{\bm{E}}{^{A_3}} \wedge \dots \wedge \tensor{\bm{\tilde{E}}}{^{A_D}} &=& (D - 3)! \Big( M \det{(\tilde{e})} \big( [\tilde{e}^{-1} e]^3 - 3 [\tilde{e}^{-1} e] [(\tilde{e}^{-1} e)^2] \\
    && + \, 2 [(\tilde{e}^{-1} e)^3] \big) + 3 N \det{(\tilde{e})} ([\tilde{e}^{-1} e]^2 - [(\tilde{e}^{-1} e)^2]) \Big), \\
    &\vdots& \\
    \hat{\epsilon}_{A_1 A_2 \cdots A_D} \tensor{\bm{\tilde{E}}}{^{A_1}} \wedge \tensor{\bm{\tilde{E}}}{^{A_2}} \wedge \tensor{\bm{E}}{^{A_{3}}} \wedge \dots \wedge \tensor{\bm{E}}{^{A_D}} &=& (D - 2)! \Big( N \det{(e)} \left( [e^{-1} \tilde{e}]^2 - [(e^{-1} \tilde{e})^2] \right) + 2 M \det{(e)} [e^{-1} \tilde{e}] \Big), \\
    \hat{\epsilon}_{A_1 A_2 \cdots A_D} \tensor{\bm{\tilde{E}}}{^{A_1}} \wedge \tensor{\bm{E}}{^{A_2}} \wedge \tensor{\bm{E}}{^{A_{3}}} \wedge \dots \wedge \tensor{\bm{E}}{^{A_D}} &=& (D - 1)! \left( N \det{(e)} [e^{-1} \tilde{e}] + M \det{(e)} \right), \\
    \hat{\epsilon}_{A_1 A_2 \cdots A_D} \tensor{\bm{E}}{^{A_1}} \wedge \tensor{\bm{E}}{^{A_2}} \wedge \tensor{\bm{E}}{^{A_3}} \wedge \dots \wedge \tensor{\bm{E}}{^{A_D}} &=& D! \, N \det{(e)},
\end{eqnarray*}
where $[\cdots]$ denotes the trace. Now we define $\tensor{X}{^I_J} \equiv \tensor{\tilde{e}}{^I_k} \tensor{e}{^k_J}$ and $\tensor{Y}{^I_J} \equiv \tensor{e}{^I_k} \tensor{\tilde{e}}{^k_J}$ such that
\begin{eqnarray}
    && \sum_{n = 0}^D \frac{c_n}{n! (D - n)!} \hat{\epsilon}_{A_1 A_2 \cdots A_D} \tensor{\bm{E}}{^{A_1}} \wedge \dots \wedge \tensor{\bm{E}}{^{A_n}} \wedge \tensor{\bm{\tilde{E}}}{^{A_{n + 1}}} \wedge \dots \wedge \tensor{\bm{\tilde{E}}}{^{A_D}} \nonumber \\
    && \qquad \qquad = d^D x \, N \det{(e)} \bigg( | \det{(X)} | \frac{M}{N} \sum_{n = 0}^{D - 1} c_n \mathcal{S}_n(Y) + \sum_{n = 0}^{D - 1} c_{D - n} \mathcal{S}_n(X) \bigg).
\end{eqnarray}
\end{widetext}
with $\mathcal{S}_n$ symmetric polynomial. For example given matrix $\mathrm{A}$ size $D \times D$,
\begin{eqnarray*}
    \mathcal{S}_0(\mathrm{A}) &=& 1, \\
    \mathcal{S}_1(\mathrm{A}) &=& [\mathrm{A}], \\
    \mathcal{S}_2(\mathrm{A}) &=& \frac{1}{2!} \Big( [\mathrm{A}]^2 - [\mathrm{A}^2] \Big), \\
    \mathcal{S}_3(\mathrm{A}) &=& \frac{1}{3!} \Big( [\mathrm{A}]^3 - 3 [\mathrm{A}] [\mathrm{A}^2] + 2 [\mathrm{A}^3] \Big), \\
    \mathcal{S}_4(\mathrm{A}) &=& \frac{1}{4!} \Big( [\mathrm{A}]^4 - 6 [\mathrm{A}]^2 [\mathrm{A}^2] + 8 [\mathrm{A}] [\mathrm{A}^3] \\
    && + \, 3 [\mathrm{A}^2]^2 - 6 [\mathrm{A}^4] \Big) \\
    &\vdots& \\
    \mathcal{S}_D(\mathrm{A}) &=& \det{(\mathrm{A})}, \\
    \mathcal{S}_{n > D}(\mathrm{A}) &=& 0.
\end{eqnarray*}

\section{MTMVMG Action}

\label{sec:appendixB}

In this Appendix we fill some gaps in the derivation of the MTMVMG action in Sec.~\ref{sec:MVMGS}. Let us begin with the MTMVMG Hamiltonian \eqref{eq:hmtmg} which can be expressed in terms of the time derivative of spatial vielbein and the scalar field,
\begin{eqnarray}
    \tensor{\dot{e}}{^I_i} &\approx& \frac{\delta \mathcal{H}_\text{MTMVMG}}{\delta \tensor{\pi}{_I^i}} \nonumber \\
    &=& \frac{1}{M_{\text{Pl}}^{D - 2}} \Bigg[ \frac{N}{\det{(e)}} \left( \tensor{\pi}{^I_i} - \frac{1}{D - 2} \tensor{\pi}{^k_K} \tensor{e}{^K_k} \tensor{e}{^I_i} \right) \nonumber \\
    && + \, (\nabla_i N_j) \delta^{I J} \tensor{e}{^j_J} \nonumber \\
    && + \, | \det{(X)} | \lambda M W(\psi) \left( \gamma_{i k} \tensor{\tilde{e}}{^k_L} \delta^{I K} - \frac{1}{D - 2} \tensor{Y}{^K_L} \tensor{e}{^I_i} \right) \nonumber \\
    && \times \, \sum_{n = 1}^{D - 1} \sum_{m = 1}^n (-1)^m c_n \tensor{\left( Y^{m - 1} \right)}{^L_K} \mathcal{S}_{n - m}(Y) \Bigg], \\
    \dot{\psi} &\approx& \frac{\delta \mathcal{H}_\text{MTMVMG}}{\delta \pi} \nonumber \\
    &=& \frac{N \pi}{\det{(e)}} + N^i \partial_i \psi \nonumber \\
    && + \, | \det{(X)} | \lambda M \frac{dW}{d\psi} \sum_{n = 1}^{D - 1} c_n \mathcal{S}_n(Y),
\end{eqnarray}
respectively, such that we have the new conjugate momenta
\begin{eqnarray}
    \frac{\tensor{\pi}{^i_I}}{\det{(e)}} &\equiv& M_{\text{Pl}}^{D - 2}(K^{i j} \delta_{I J} \tensor{e}{^J_j} - K \tensor{e}{^i_I}) \nonumber \\
    && - \, \lambda \frac{W(\psi)}{2} \frac{M}{N} \Theta^{i j} \delta_{I J} \tensor{e}{^J_j}, \label{eq:conjugatemomentum1} \\
    \frac{\pi}{\det{(e)}} &\equiv& \frac{\dot{\psi}}{N} - \frac{N^i}{N} \partial_i \psi - \lambda \frac{dW}{d\psi} \frac{M}{N} \Phi. \label{eq:conjugatemomentum2}
\end{eqnarray}
The last terms in Eqs.~\eqref{eq:conjugatemomentum1} and \eqref{eq:conjugatemomentum2} arise since the MTMVMG constraints depend on the conjugate momenta $\tensor{\pi}{_I^i}$ and $\pi$. The symmetrical property of $\Theta^{ij}$ cancels the antisymmetric terms, namely $\alpha_{M N} \mathcal{P}^{[M N]}$ and $\beta_{M N} Y^{[M N]}$ out of the MTMVMG Hamiltonian. Employing the Legendre transformation, we obtain the MTMVMG action,
\begin{widetext}
\begin{eqnarray}
    S_{\text{MTMVMG}} &\equiv& \int_{\mathcal{M}} d^D x \, \Big( \tensor{\pi}{^i_I} \tensor{\dot{e}}{^I_i} + \pi \dot{\psi} - \mathcal{H}_\text{MTMVMG} \big|_{\alpha_{M N} = \beta_{M N} = 0} \Big) \\
    &=& S_{\text{pre}} - \int_{\mathcal{M}} d^D x (\lambda \mathcal{C}_0 + \lambda^i \mathcal{C}_i) - \frac{2}{M_{\text{Pl}}^{D - 2}} \int_{\mathcal{M}} d^D x \, N \det{(e)} \left( \lambda \frac{W(\psi)}{4} \frac{M}{N} \right)^2 \left( \Theta^{i j} \Theta_{i j} - \frac{1}{D - 2} \Theta^2 \right) \nonumber \\
    && - \, \frac{1}{2} \int_{\mathcal{M}} d^D x \, N \det{(e)} \left( \lambda \frac{dW}{d\psi} \frac{M}{N} \right)^2 \Phi^2 + S_\text{matter} \\
    &=& S_{\text{pre}} - \int_{\mathcal{M}} d^D x \, (\lambda \mathcal{\bar{C}}_0 + \lambda^i \mathcal{C}_i) + \frac{2}{M_{\text{Pl}}^{D - 2}} \int_{\mathcal{M}} d^D x \, N \det{(e)} \left( \lambda \frac{W(\psi)}{4} \frac{M}{N} \right)^2 \left( \Theta^{i j} \Theta_{i j} - \frac{1}{D - 2} \Theta^2 \right) \nonumber \\
    && + \, \frac{1}{2} \int_{\mathcal{M}} d^D x \, N \det{(e)} \left( \lambda \frac{dW}{d\psi} \frac{M}{N} \right)^2 \Phi^2 + S_\text{matter},
\end{eqnarray}
where $S_{\text{pre}}$ and $S_\text{matter}$ are the precursor action and the matter action, respectively, whereas
\begin{eqnarray}
    \mathcal{C}_0 &=& W(\psi) M \det{(e)} | \det{(X)} | \sum_{n = 1}^{D - 1} \sum_{m = 1}^n (-1)^m c_n \tensor{\left( Y^{m - 1} \right)}{^J_I} \mathcal{S}_{n - m}(Y) \nonumber \\
    && \times \left[ \left( \gamma_{i k} \tensor{\tilde{e}}{^k_J} \tensor{e}{^I_j} - \frac{1}{D - 2} \gamma_{i j} \tensor{Y}{^I_J} \right) \left( K^{i j} - K \gamma^{i j} - \frac{\lambda}{M_{\text{Pl}}^{D - 2}} \frac{W(\psi)}{2} \frac{M}{N} \Theta^{i j} \right) - \frac{1}{M} \tensor{\tilde{e}}{^k_J} \frac{\partial}{\partial t} \tensor{\tilde{e}}{^I_k} \right] \nonumber \\
    && + \, M \det{(e)} \frac{dW}{d\psi} \left( \frac{\dot{\psi}}{N} - \frac{N^i}{N} \partial_i \psi \right) \Phi - \lambda N \det{(e)} \left( \frac{dW}{d\psi} \frac{M}{N} \right)^2 \Phi^2, \\
    \mathcal{C}_i &=& M \det{(e)} \bigg[ W(\psi) \nabla^j \bigg( | \det{(X)} | \sum_{n = 1}^{D - 1} \sum_{m = 1}^n (-1)^m c_n \tensor{\left( Y^{m - 1} \right)}{^J_I} \mathcal{S}_{n - m}(Y) M \tensor{Y}{^K_J} \delta_{K L} \tensor{e}{^I_i} \tensor{e}{^L_j} \bigg) \nonumber \\
    && + \, | \det{(X)} | \partial_i \psi \frac{dW}{d\psi} \sum_{n = 1}^{D - 1} c_n \mathcal{S}_n(Y) \bigg],
\end{eqnarray}
with
\begin{eqnarray}
    \mathcal{\bar{C}}_0 \equiv \mathcal{C}_0 \big|_{\lambda = 0} &=& M | \det{(X)} | \Bigg\{ W(\psi) \left[ \left( \gamma_{i k} \tensor{\tilde{e}}{^k_J} \tensor{e}{^I_j} - \frac{1}{D - 2} \gamma_{i j} \tensor{Y}{^I_J} \right) (K^{i j} - K \gamma^{i j}) - \frac{1}{M} \tensor{\tilde{e}}{^k_J} \frac{\partial}{\partial t} \tensor{\tilde{e}}{^I_k} \right] \nonumber \\
    && \times \sum_{n = 1}^{D - 1} \sum_{m = 1}^n (-1)^m c_n \tensor{\left( Y^{m - 1} \right)}{^J_I} \mathcal{S}_{n - m}(Y) \nonumber \\
    && - \, \frac{dW}{d\psi} \left( \frac{\dot{\psi}}{N} - \frac{N^i}{N} \partial_i \psi \right) \sum_{n = 1}^{D - 1} c_n \mathcal{S}_n(Y) \Bigg\}.
\end{eqnarray}

\section{Linear Stability}

\label{sec:appendixC}

\begin{table}[htbp]

\small

\begin{tabularx}{\columnwidth}{| >{\centering\arraybackslash}>{\hsize=1.0\hsize}X | >{\centering\arraybackslash}>{\hsize=1.0\hsize}X |}

\hline

Eigenvalues & Stability \\

\hline
\hline

\multicolumn{2}{| c |}{Real eigenvalues} \\

\hline
\hline

$\mu_n < 0$, for $n = 1, 2, 3$ & stable node \\

\hline

$\mu_n > 0$, for $n = 1, 2, 3$ & unstable node \\

\hline

$\mu_n > 0$ and $\mu_m \leq 0$, for $n \cup m = 1, 2, 3$ & unstable \\

\hline

$\mu_n > 0$ and $\mu_m < 0$, for $n \cup m = 1, 2, 3$ & saddle point \\

\hline

$\mu_n = 0$ and $\mu_m < 0$, for $n \cup m = 1, 2, 3$ & nonhyperbolic, linear stability fails to determine and other methods are needed \\

\hline
\hline

\multicolumn{2}{| c |}{Complex eigenvalues} \\

\hline
\hline

$\text{Re}(\mu_n) < 0$, for $n = 1, 2, 3$ & strongly stable spiral \\

\hline

$\text{Re}(\mu_n) > 0$, for $n = 1, 2, 3$ & strongly unstable spiral \\

\hline

$\text{Re}(\mu_n) < 0$ and $\text{Re}(\mu_m) > 0$, for $n \cup m = 1, 2, 3$ & saddle focus \\

\hline

$\text{Re}(\mu_n) = 0$ and $\text{Re}(\mu_m) \ne 0$, for $n \cup m = 1, 2, 3$ & weakly center point \\

\hline

$\text{Re}(\mu_n) = 0$, for $n = 1, 2, 3$ & strongly center point \\

\hline
\hline

\multicolumn{2}{| c |}{Real eigenvalues $\mu_n$ and complex eigenvalues $\mu_m$ for $n \cup m = 1, 2, 3$} \\

\hline
\hline

$\mu_n < 0$ and $\text{Re}(\mu_m) < 0$ & weakly stable spiral \\

\hline

$\mu_n > 0$ and $\text{Re}(\mu_m) > 0$ & weakly unstable spiral \\

\hline

$\mu_n < 0$ and $\text{Re}(\mu_m) > 0$ & saddle focus \\

\hline

$\mu_n > 0$ and $\text{Re}(\mu_m) < 0$ & saddle focus \\

\hline

\end{tabularx}

\caption{\label{tab:criticalpoints} Stability properties of the critical point $(x_{\psi, c}, x_{V, c}, x_{W, c})$ based on the three eigenvalues $\mu_1, \mu_2, $ and $\mu_3$.}

\end{table}

In this Appendix we consider the linear perturbation of dynamical equations \eqref{eq:xpsiprime}-\eqref{eq:xwprime} in the critical points $(x_{\psi, c}, x_{V, c}, x_{W, c})$. First, we expand autonomous variables around these points
\begin{eqnarray}
    x_\psi &=& x_{\psi, c} + u_\psi, \\
    x_V &=& x_{V, c} + u_V, \\
    x_W &=& x_{W, c} + u_W.
\end{eqnarray}
The first order equation of motions has the form
\begin{eqnarray}
    \frac{2 u_\psi'}{D - 1} &=& \Big[ 3 (1 - w_m) x_{\psi, c}^2 - (1 + w_m) x_{V, c}^2 - (f_2(u) + w_m f_1(u)) x_{W, c}^2 - (1 - w_m) \Big] u_\psi \nonumber \\
    && - \, \Bigg[ 2 (1 + w_m) x_{\psi, c} - 4 \sqrt{\frac{D - 2}{D - 1}} \lambda_V \Bigg] x_{V, c} u_V - \Bigg[ 2 (f_2(u) + w_m f_1(u)) x_{\psi, c} - 4 \sqrt{\frac{D - 2}{D - 1}} \lambda_W f_1(u) \Bigg] x_{W, c} u_W, \nonumber \\ \\
    \frac{2 u_V'}{D - 1} &=& \Bigg[ 2 (1 - w_m) x_{\psi, c} - \sqrt{\frac{D - 2}{D - 1}} \lambda_V \Bigg] x_{V, c} u_\psi + \Bigg[ (1 + w_m) + (1 - w_m) x_{\psi, c}^2 \nonumber \\
    && - \, 3 (1 + w_m) x_{V, c}^2 - (f_2(u) + w_m f_1(u)) x_{W, c}^2 - \sqrt{\frac{D - 2}{D - 1}} \lambda_V x_{\psi, c} \Bigg] u_V - 2 (f_2(u) + w_m f_1(u)) x_{V, c} x_{W, c} u_W, \nonumber \\ \\
    \frac{2 u_W'}{D - 1} &=& \Bigg[ 2 (1 - w_m) x_{\psi, c} - \sqrt{\frac{D - 2}{D - 1}} \lambda_W \Bigg] x_{W, c} u_\psi - 2 (1 + w_m) x_{V, c} x_{W, c} u_V \nonumber \\
    && + \, \Bigg[ (1 + w_m) + (1 - w_m) x_{\psi, c}^2 - (1 + w_m) x_{V, c}^2 - 3 (f_2(u) + w_m f_1(u)) x_{W, c}^2 - \sqrt{\frac{D - 2}{D - 1}} \lambda_W x_{\psi, c} \Bigg] u_W,
\end{eqnarray}
such that we can cast the above equations into the matrix form
\begin{equation}
    \begin{pmatrix} u_\psi' \\ u_V' \\ u_W' \end{pmatrix} = \bm{J} \begin{pmatrix} u_\psi \\ u_V \\ u_W \end{pmatrix},
\end{equation}
where $\bm{J}$ is Jacobian matrix. Let the eigenvalues of $\bm{J}$ are $\mu_1$, $\mu_2$, and $\mu_3$. We write the Jacobian matrix $\bm{J}$ and its eigenvalues for each of the critical points below. These then can be used to analyze their stability properties, which we summarize in Table \ref{tab:criticalpoints}.
\begin{itemize}

\item The Jacobian matrix for CP$_1$: $(0, 0, 0)$ is
\begin{equation}
    \bm{J} = \begin{pmatrix}
    -(1 - w_m) & 0 & 0 \\
    0 & 1 + w_m & 0 \\
    0 & 0 & 1 + w_m
    \end{pmatrix},
\end{equation}
with eigenvalues
\begin{equation}
    \{ -(1 - w_m), 1 + w_m \}.
\end{equation}

\item The Jacobian matrix for CP$_2$: $(1, 0, 0)$ is
\begin{equation}
    \bm{J} = \begin{pmatrix}
    2 (1 - w_m) & 0 & 0 \\
    0 & 2 - \lambda_V \sqrt{\frac{D - 2}{D - 1}} & 0 \\
    0 & 0 & 2 - \lambda_W \sqrt{\frac{D - 2}{D - 1}}
    \end{pmatrix},
\end{equation}
with eigenvalues
\begin{equation}
    \Bigg\{ 2 (1 - w_m), 2 - \lambda_V \sqrt{\frac{D - 2}{D - 1}}, 2 - \lambda_W \sqrt{\frac{D - 2}{D - 1}} \Bigg\}.
\end{equation}

\item The Jacobian matrix for CP$_3$: $(-1, 0, 0)$ is
\begin{equation}
    \bm{J} = \begin{pmatrix}
    2 (1 - w_m) & 0 & 0 \\
    0 & 2 + \lambda_V \sqrt{\frac{D - 2}{D - 1}} & 0 \\
    0 & 0 & 2 + \lambda_W \sqrt{\frac{D - 2}{D - 1}}
    \end{pmatrix},
\end{equation}
with eigenvalues
\begin{equation}
    \Bigg\{ 2 (1 - w_m), 2 + \sqrt{\frac{D - 2}{D - 1}}, 2 + \lambda_W \sqrt{\frac{D - 2}{D - 1}} \Bigg\}.
\end{equation}

\item The Jacobian matrix for CP$_4$: $(x_{\psi, c}, 0, 0)$ is
\begin{equation}
    \bm{J} = \begin{pmatrix}
    0 & 0 & 0 \\
    0 & 2 - \lambda_V x_{\psi, c} \sqrt{\frac{D - 2}{D - 1}} & 0 \\
    0 & 0 & 2 - \lambda_W x_{\psi, c} \sqrt{\frac{D - 2}{D - 1}}
    \end{pmatrix},
\end{equation}
with eigenvalues
\begin{equation}
    \Bigg\{ 0, 2 - \lambda_V x_{\psi, c} \sqrt{\frac{D - 2}{D - 1}}, 2 - \lambda_W x_{\psi, c} \sqrt{\frac{D - 2}{D - 1}} \Bigg\}.
\end{equation}

\item One of the eigenvalues of the Jacobian matrix for CP$_5$: $\bigg( \frac{\mathcal{A}_\pm}{\lambda_V} \sqrt{\frac{D - 1}{D - 2}}, \frac{1}{\lambda_V} \sqrt{\frac{(D - 1)(2 - \mathcal{A}_\pm) \mathcal{A}_\pm}{2 (D - 2)}}, 0 \bigg)$ is
\begin{equation}
    \mu_1 = \left( 1 - \frac{\lambda_W}{\lambda_V} \right) \mathcal{A}_\pm.
\end{equation}

\item The Jacobian matrix for CP$_6$: $\bigg( 0, \sqrt{\frac{\lambda_W}{\lambda_W - \lambda_V}}, \sqrt{\frac{\lambda_V}{| f_1(u) | (\lambda_W - \lambda_V)}} \bigg)$ is
\begin{equation}
    \bm{J} = \begin{pmatrix}
    -2 & -4 \lambda_V \sqrt{\frac{D - 2}{D - 1} \frac{\lambda_W}{\lambda_W - \lambda_V}} & 4 \lambda_W \sqrt{\frac{D - 2}{D - 1} \frac{| f_1(u) | \lambda_V}{\lambda_W - \lambda_V}} \\
    -\lambda_V \sqrt{\frac{D - 2}{D - 1} \frac{\lambda_W}{\lambda_W - \lambda_V}} & -\frac{2 (1 + w_m) \lambda_W}{\lambda_W - \lambda_V} & -\frac{2 (1 + w_m) \sqrt{| f_1(u) | \lambda_V \lambda_W}}{\lambda_W - \lambda_V} \\
    -\lambda_W \sqrt{\frac{D - 2}{D - 1} \frac{\lambda_V}{| f_1(u) | (\lambda_W - \lambda_V)}} & -\frac{2 (1 + w_m)}{\lambda_W - \lambda_V} \sqrt{\frac{\lambda_V \lambda_W}{| f_1(u) |}} & \frac{2 (1 + w_m) \lambda_V}{\lambda_W - \lambda_V}
    \end{pmatrix},
\end{equation}
with eigenvalues
\begin{equation}
    \Bigg\{ -2 (1 + w_m), -1 \pm \sqrt{1 - \frac{4 (D - 1) \lambda_V \lambda_W}{D - 2}} \Bigg\}.
\end{equation}

\item The Jacobian matrix for CP$_7$: $\bigg( 0, \sqrt{1 - f_1(u) x_{W, c}^2}, x_{W, c} \bigg)$ is
\begin{equation}
    \bm{J} = \begin{pmatrix}
    -2 & 0 & 0 \\
    0 & -2 (1 + w_m) (1 - f_1(u) x_{W, c}^2) & -2 (1 + w_m) f_1(u) x_{W, c} \sqrt{1 - f_1(u) x_{W, c}^2} \\
    0 & -2 (1 + w_m) x_{W, c} \sqrt{1 - f_1(u) x_{W, c}^2} & -2 (1 + w_m) f_1(u) x_{W, x}^2
    \end{pmatrix},
\end{equation}
with eigenvalues
\begin{equation}
    \{ -2, -2 (1 + w_m), 0 \}.
\end{equation}

\item One of the eigenvalues of the Jacobian matrix for CP$_8$: $\bigg( \frac{\mathcal{B}_\pm}{\lambda_W} \sqrt{\frac{D - 1}{D - 2}}, 0, \frac{1}{\lambda_W} \sqrt{\frac{(D - 1)(2 - \mathcal{B}_\pm) \mathcal{B}_\pm}{2 (D - 2) f_1(u)}} \bigg)$ is
\begin{equation}
    \mu_1 = \left( 1 - \frac{\lambda_V}{\lambda_W} \right) \mathcal{B}_\pm.
\end{equation}

\end{itemize}
\end{widetext}

\bibliographystyle{apsrev4-1}
\bibliography{references}

\end{document}